\documentclass{article}

\usepackage{hyperref}
\usepackage{amsmath}
\usepackage{amssymb}
\usepackage{amsthm}
\usepackage{tikz}
\usepackage{stmaryrd}
\usepackage[ruled,linesnumbered]{algorithm2e}
\newcounter{global}
\theoremstyle{definition}
\newtheorem{definition}[global]{Definition}
\theoremstyle{plain}
\newtheorem{theorem}[global]{Theorem}
\newtheorem{lemma}[global]{Lemma}
\newtheorem{proposition}[global]{Proposition}
\newtheorem{corollary}[global]{Corollary}
\newtheoremstyle{note}{}{}{}{}{\itshape}{.}{.5em}{}
\theoremstyle{note}
\newtheorem{remark}{Remark}%
\newtheorem{example}{Example}%
\renewcommand\section{%
	\@startsection {section}{1}{\z@}%
	{-3.5ex \@plus -1ex \@minus -.2ex}%
	{2.3ex \@plus.2ex}%
	{\normalfont\large\bfseries}}
\def\nat{\mathbb}
\def\Z{\nat Z}
\def\alltatr{\mathcal{T}_Y}

\def\|{\,|\,}
\def\nmodels{\not\models}
\def\proves{\vdash}
\def\nproves{\nvdash}
\def\itm#1{{\rm(\textit{\romannumeral#1})}}
\def\mdl#1{\triangle^{\!#1}}
\def\logand{\mathop{\binampersand}}

\def\derfrac#1#2{\ensuremath{\displaystyle%
		\cfrac{#1}{#2\rule{0pt}{0.9em}}}}
\def\dderfrac#1#2{\ensuremath{\displaystyle%
		\cfrac{#1}{#2\rule{0pt}{0.9em}}}}

\def\PL{\ensuremath{\mathrm{PL}}}

\def\lfrac#1#2#3{\derfrac{#1}{#2}\,{\scriptstyle #3}}
\def\llfrac#1#2#3{\dderfrac{#1}{#2}\,{\scriptstyle #3}}

\def\Ax{\ensuremath{\mathop{(\mathrm{Ax})}}}
\def\Cut{\ensuremath{\mathop{(\mathrm{Cut})}}}
\def\Sim{\ensuremath{\mathop{(\mathrm{Sim})}}}
\def\iCut{\ensuremath{\mathop{(\mathrm{Cut}_i)}}}
\def\iSim{\ensuremath{\mathop{(\mathrm{Sim}_i)}}}
\def\Shf{\ensuremath{\mathop{(\mathrm{Shf})}}}
\def\Add{\ensuremath{\mathop{(\mathrm{Add})}}}
\def\Aug{\ensuremath{\mathop{(\mathrm{Aug})}}}
\def\Pro{\ensuremath{\mathop{(\mathrm{Pro})}}}
\def\Tra{\ensuremath{\mathop{(\mathrm{Tra})}}}
\def\Acc{\ensuremath{\mathop{(\mathrm{Acc})}}}
\def\Ref{\ensuremath{\mathop{(\mathrm{Ref})}}}
\def\Wea{\ensuremath{\mathop{(\mathrm{Wea})}}}

\def\AX#1{\ensuremath{\llfrac{}{#1}{\Ax}}}
\def\CUT#1#2#3{\ensuremath{\lfrac{#1,\ #2}{#3}{\Cut}}}
\def\SHF#1#2{\ensuremath{\lfrac{#1}{#2}{\Shf}}}
\def\iCUT#1#2#3{\ensuremath{\lfrac{#1,\ #2}{#3}{\iCut}}}
\def\iSIM#1#2#3{\ensuremath{\lfrac{#1,\ #2}{#3}{\iSim}}}
\def\SIM#1#2#3{\ensuremath{\lfrac{#1,\ #2}{#3}{\Sim}}}

\DeclareMathSymbol{\Sigma}{\mathalpha}{operators}{"06}

\begin{document}
	
\title{Logic of temporal attribute implications}
	
\date{\normalsize%
  Dept. Computer Science, Palacky University Olomouc}
	
\author{Jan Triska\footnote{%
    e-mail: \texttt{jan.triska@upol.cz}, phone: +420 585 634 715, fax:
    +420 585 411 643%
    % \newline
    % address: Palacky University,
    % 17. listopadu 12, CZ--77146 Olomouc, Czech Republic
  } \and Vilem Vychodil}
	
\maketitle
	
\begin{abstract}
  We study logic for reasoning with if-then formulas describing
  dependencies between attributes of objects which are observed in
  consecutive points in time. We introduce semantic entailment of the
  formulas, show its fixed-point characterization, investigate closure
  properties of model classes, present an axiomatization and prove its
  completeness, and investigate alternative axiomatizations and
  normalized proofs. We investigate decidability and complexity
  issues of the logic and prove that the entailment problem 
  is NP-hard and belongs to EXPSPACE. We show that by restricting
  to predictive formulas, the entailment problem is decidable
  in pseudo-linear time.
		
  \medskip\noindent%
\end{abstract}
	
%%%%%%%%%%%%%%%%%%%%%%%%%%%%%%%%%%%%%%%%%%%%%%%%%%%%%%%%%%%%%%%%%%%%%%%%%%%%%%%%
%%%%% INTRODUCTION
%%%%%%%%%%%%%%%%%%%%%%%%%%%%%%%%%%%%%%%%%%%%%%%%%%%%%%%%%%%%%%%%%%%%%%%%%%%%%%%%
\section{Introduction}
Formulas describing if-then dependencies between attributes play
fundamental role in reasoning about attributes in many disciplines
including database systems~\cite{Co:Armodflsdb,Mai:TRD}, formal
concept analysis~\cite{GaWi:FCA,GuDu}, data
mining~\cite{AgImSw:ASR,Zak:Mnrar}, logic programming
\cite{Lloyd84,Ro:Molbrp}, and their applications.  In these
disciplines, the rules often appear under different names (e.g.,
attribute implications, functional dependencies, or simply ``rules'')
with semantics defined in various structures (e.g., transactional
data, boolean matrices, or $n$-ary relations) but as it has been shown
in~\cite{Fagin}, the rules may be seen as propositional formulas with
the semantic entailment defined as in the propositional logic,
possibly extended by additional measures of interestingness. The rules
are popular because of their easy readability for non-expert users and
tractability of the entailment problem which is decidable in linear
time~\cite{BeBe:Cprttdonfrs}. Research on if-then rules is active and
recent results include new theoretical
observations~\cite{BaKaNa:Cfdfcapa, FeHaLi:Raffhdopr, Ib:OfdqHt, MaFaBr:Eidc,
Ol:raathfd,SoChCh:Eddtdd} on the rules and their generalizations as well as
applications in data analysis, formal languages, and related
areas~\cite{CoSa:Ibtfdsav, DECaBr:Dwfdmg, DiJeLaSp:Mfcqufid, FaLiTaYu:Ididd,
LiLiToYo:Epdcfd,LiYeLiWa:Odfdfd,SzBe:Fdssgecfl,ViLiMo:TipfcfdcXd,We:Efdiqpr}.
	
In this paper, we introduce if-then formulas which express presence of
attributes relatively in time and the formulas are evaluated in data
where the presence or absence of attributes changes in time. In our
approach, we adopt the notion of a discrete time, i.e., the data are
observed at distinct points in time. We consider a formula valid in
data changing over time if the if-then dependency prescribed by the
formula holds in all time points. We introduce the formulas as
expressions
\begin{align}
  \bigl\{y_1^{i_1},\ldots,y_m^{i_m}\bigr\}
  \Rightarrow
  \bigl\{z_1^{j_1},\ldots,z_n^{j_n}\bigr\},
  \label{eqn:cimpl}
\end{align}
where $y_1,\ldots,y_m$ and $z_1,\ldots,z_n$ are attributes which may
be viewed as propositional variables, and
$i_1,\ldots,i_m,j_1,\ldots,j_n$ are integers annotating the attributes
by \emph{relative time points} with the following meaning: $0$ denotes
the present time point, $1$ is its immediate successor, $-1$ is the
immediate predecessor of $0$, $2$ is the immediate successor of $1$,
etc. With this interpretation of time points and considering, for
instance, the unit of time ``a day'', formula
$\{x^{-1},y^{0}\} \Rightarrow \{z^1\}$ prescribes the following
dependency: ``If $x$ was present yesterday and $y$ is present today,
then $z$ will be present tomorrow.''  From our perspective, a classic
if-then formula
\begin{align}
  \bigl\{y_1,\ldots,y_m\bigr\}
  \Rightarrow
  \bigl\{z_1,\ldots,z_n\bigr\},
  \label{eqn:impl}
\end{align}
may be seen as a particular case of~\eqref{eqn:cimpl}, where all the
relative time points $i_1,\ldots,i_m,j_1,\ldots,j_n$ are equal to $0$,
and the data in which the formula is evaluated is constant in all time
points.
	
We provide answers to several questions which emerge with formulas
like~\eqref{eqn:cimpl}. First, we define the notion of semantic
entailment of the formulas, investigate closure structures of models
of theories consisting of such formulas, and show that the problem of
checking whether a formula is semantically entailed by a set of
formulas can be reduced to checking its validity in a single model.
Second, we prove that the semantic entailment has a complete
axiomatization. That is, we show a notion of provability of formulas
like~\eqref{eqn:cimpl} and show that it coincides with the semantic
entailment. We discuss several possible axiomatizations, including
ones that can be used to consider proofs in particular normal
forms. Third, based on our insight into the properties of the semantic
entailment and provability, we derive results on decidability and
complexity of the entailment problem. Fourth, we include notes on the
relationship of the formulas to formulas appearing in modal
logics~\cite{BlRiVe:ML} and triadic formal concept
analysis~\cite{LeWi:Tafca}. Similar rules as we consider in this paper
appeared as inter-transaction association rules~\cite{TLHF:bbtmiar}
inferred from time-changing transactional data.  Despite the
popularity of the rules in data mining, a logical analysis of the
entailment of the rules and related properties is missing---providing
the logical foundations is a goal of our paper.
	
Our paper is organized as follows. In Section~\ref{sec:related} and
Section~\ref{sec:prelim}, we present a survey of related work and
short preliminaries. We introduce the formulas and present the results
on their semantic entailment in Section~\ref{sec:sem}.
In Section \ref{sec:compl} we give complete
axiomatizations and in Section~\ref{sec:comput} we deal with the
related computational issues. Finally, in Section~\ref{sec:concl}, we
present a conclusion.
	
%%%%%%%%%%%%%%%%%%%%%%%%%%%%%%%%%%%%%%%%%%%%%%%%%%%%%%%%%%%%%%%%%%%%%%%%%%%%%%%%
%%%%% RELATED WORK
%%%%%%%%%%%%%%%%%%%%%%%%%%%%%%%%%%%%%%%%%%%%%%%%%%%%%%%%%%%%%%%%%%%%%%%%%%%%%%%%
\section{Related Work}\label{sec:related}
In database systems and knowledge engineering, there appeared isolated
approaches which propose temporal semantics of if-then rules. We
present here a short survey of the approaches and highlight the
differences between our approach and the existing ones.
	
Formulas called temporal functional dependencies emerged in databases
with time granularities~\cite{BeJaWa:TGD}. In this approach, a time
granularity is a general partition of time like seconds, weeks, years,
etc., and a time granularity is associated to each relational
schema. In addition, each tuple in a relation is associated with a
part (so-called granule) of granularity. In this setting, temporal
functional dependencies are like the ordinary functional
dependencies~\cite{Fagin} with a time granularity as an additional
component. The concept of validity of temporal functional dependencies
is defined in much the same way as its classic counterpart and
includes an additional condition that granules of tuples need to be
covered by any granule from granularity of the temporal functional
dependency.  Thus, \cite{BeJaWa:TGD} uses an ordinary notion of
validity of functional dependencies which is restricted to some time
segments. This is conceptually very different from the problem we deal
with in this paper.
	
Several approaches to temporal if-then rules, which are conceptually
similar to~\cite{BeJaWa:TGD}, appeared in the field of association
rules~\cite{AgImSw:ASR,Zak:Mnrar} as the so-called temporal
association rules~\cite{AlRo:Adtar,LiNiWaJa:Dcbtar,RaRo:PDMKD}. In
these approaches, the input data is in the form of transactions (i.e.,
subsets of items) where each transaction occurred at some point in
time and the interest of the papers lies in extracting association
rules from data which occur during a specified time cycle. For
instance, one may be interested in extracting rules which are valid in
``every spring month of a year'', ``every Monday in every year'',
etc. As in the case of the temporal functional dependencies, the
temporal association rules may be understood as classic association
rules occurring during specified time cycles.
	
Other results motivated by temporal semantics of association rules
includes the so-called inter-transaction association
rules~\cite{FDL:Iarmcptasmd,FYLH:tmmiar,LWWCW:aeamcii,TLHF:bbtmiar},
see \cite{LFH:biaammiar} for a survey of approaches. The papers
propose algorithms to extract, given an input transactional data and a
measure of interestingness (based on levels of minimal support and
confidence), if-then rules which are preserved over a given period of
time.  From this point of view, the rules can be seen as formulas
studied in this paper restricted to so-called predictive rules (see
Definition~\ref{def:reg} in Section~\ref{sec:comput}) whose validity
is considered with respect to the additional parameter of
interestingness. As a consequence, the inter-transaction association
rules are related to the rules in our approach in the same way as the
ordinary association rules~\cite{AgImSw:ASR} are related to the
ordinary attribute implications~\cite{GaWi:FCA}. The results in
\cite{FDL:Iarmcptasmd,FYLH:tmmiar,LWWCW:aeamcii,LFH:biaammiar,TLHF:bbtmiar}
are focused almost exclusively on algorithms for mining the
inter-transaction association rules and are not concerned with
problems of entailment of the rules and the underlying logic. In
contrast, the problems of entailment of rules are central to this
paper and we show there is reasonably strong logic for reasoning with
such rules. Our observations may stimulate further development in the
field of inter-transaction association rules and similar formulas and
their applications in various
domains~\cite{FDL:Iarmcptasmd,HKS:emstarad}.

The formulas studied in this paper are also related to particular
program rules which appear in $\text{\emph{Datalog}}$ extensions
dealing with flow of time and related
phenomena~\cite{ChIm:Tddio,ChIm:Rsiqa,ChIm:friqa}
such as $\text{\emph{Datalog}}_{nS}$ ($\text{\emph{Datalog}}$
with $n$ successors). The formulas we consider in our paper correspond
to a fragment of rules which appear in such $\text{\emph{Datalog}}$
extensions. Despite the similar form of our formulas and the program rules,
there does not seem to be a direct relationship (or a reduction) of the
entailment problem of our formulas and the recognition problem
of $\text{\emph{Datalog}}_{nS}$ programs.
	
%%%%%%%%%%%%%%%%%%%%%%%%%%%%%%%%%%%%%%%%%%%%%%%%%%%%%%%%%%%%%%%%%%%%%%%%%%%%%%%%
%%%%%   PRELIMINARIES
%%%%%%%%%%%%%%%%%%%%%%%%%%%%%%%%%%%%%%%%%%%%%%%%%%%%%%%%%%%%%%%%%%%%%%%%%%%%%%%%
\section{Preliminaries}\label{sec:prelim}
In this section, we present the basic notions of closure systems (also known
as Moore families) and closure operators which are used further in the paper.
More details can be found in~\cite{Bir:LT,DaPr}.
	
If $Y$ is a set, we denote by $2^Y$ its power set. A closure operator
on $Y$ is a map $c\!: 2^Y \to 2^Y$ such that
\begin{align}
  A &\subseteq c(A),
      \label{eqn:c:ext} \\
  A \subseteq B &\text{ implies } c(A) \subseteq c(B),
                  \label{eqn:c:mon} \\
  c(c(A)) &\subseteq c(A), 
            \label{eqn:c:idm}
\end{align}
for all $A,B \subseteq Y$. The conditions
\eqref{eqn:c:ext}--\eqref{eqn:c:idm} are called the extensivity,
monotony, and idempotency of $c$, respectively.  Note
that~\eqref{eqn:c:ext} and~\eqref{eqn:c:idm} yield $c(A) = c(c(A))$
for all $A \subseteq Y$. A closure operator $c\!: 2^Y \to 2^Y$ is
called an algebraic closure operator whenever
\begin{align}
  c(A) &= \textstyle\bigcup\{c(B) \,|\,
         B \subseteq A \text{ and } B \text{ is finite}\}
\end{align}
for all $A \subseteq Y$. Moreover, $A \subseteq Y$ is called a fixed
point of $c$ whenever $c(A) = A$.
	
A system $\mathcal{S} \subseteq 2^Y$ is called a closure system on $Y$
if it is closed under arbitrary intersections, i.e.,
$\bigcap \mathcal{A} \in \mathcal{S}$ for any
$\mathcal{A} \subseteq \mathcal{S}$.  In the paper we utilize the
well-known correspondence between closure systems and closure
operators on $Y$.  In particular, if $c$ is an algebraic closure
operator on $Y$, we call the closure system of all its fixed points
the algebraic closure system induced by $c$.
	
%%%%%%%%%%%%%%%%%%%%%%%%%%%%%%%%%%%%%%%%%%%%%%%%%%%%%%%%%%%%%%%%%%%%%%%%%%%%%%%%
%%%%% FORMULAS AND SEMANTIC ENTAILMENT
%%%%%%%%%%%%%%%%%%%%%%%%%%%%%%%%%%%%%%%%%%%%%%%%%%%%%%%%%%%%%%%%%%%%%%%%%%%%%%%%
\section{Formulas, Models, and Semantic Entailment}\label{sec:sem}
In this section, we present a formalization of the formulas, their
interpretation, and semantic entailment.  Let us assume that $Y$ is a
non-empty and finite set of symbols called attributes. Furthermore, we
use integers in order to denote time points.  We put
\begin{align}
  \alltatr = \bigl\{y^i \| y \in Y \text{ and } i \in \Z\bigr\}
  \label{eqn:alltatr}
\end{align}
and interpret each $y^i \in \alltatr$ as ``attribute $y$ observed in
time $i$'' (technically, $\alltatr$ can be seen as the Cartesian
product $Y \times \Z$). Under this notation, we may now formalize
rules like~\eqref{eqn:cimpl} as follows:
	
\begin{definition}\label{def:cimpl}\upshape%
  An \emph{attribute implication over $Y$ annotated by time points} in
  $\Z$ is a formula of the form $A \Rightarrow B$, where $A,B$ are
  finite subsets of $\alltatr$.
\end{definition}
	
As we have outlined in the introduction, the purpose of time points
encoded by integers which appear in antecedents and consequents of the
considered formulas is to express points in time relatively to a
current time point. Hence, the intended meaning of \eqref{eqn:cimpl}
abbreviated by $A \Rightarrow B$ is the following: ``For all time
points $t$, if an object has all the attributes from $A$ considering
$t$ as the current time point, then it must have all the attributes
from $B$ considering $t$ as the current time point''.  In what
follows, we formalize the interpretation of $A \Rightarrow B$ in this
sense.
	
Since we wish to define formulas being true in all time points (we are
interested in formulas preserved over time), we need to shift relative
times expressed in antecedents and consequents in formulas with
respect to a changing time point. For that purpose, for each
$M \subseteq \alltatr$ and $i \in \Z$, we may introduce a subset
$M + j$ of $\alltatr$ by
\begin{align}
  M + j &= \bigl\{y^{i+j} \| y^i \in M\bigr\}
          \label{eqn:shft}
\end{align}
and call it a \emph{time shift of $M$ by $j$} (shortly, a $j$-shift of
$M$).  In the paper, we utilize the following properties of time
shifts.
	
\begin{proposition}
  For all $M,N \subseteq \alltatr$,
  $\{N_k \subseteq \alltatr \,|\, k \in K\}$, and $i,j \in \Z$, we get
  \begin{align}
    & \text{if } M \subseteq N \text{ then } M+i \subseteq N+i,
      \label{eqn:ts_mon}
    \\
    & (M+i)+j = M+(i+j),
      \label{eqn:ts_assoc}
    \\
    & \textstyle\bigcup_{k \in K} (N_k+i) = \textstyle\bigcup_{k \in K} N_k+i,
      \label{eqn:ts_cup}
    \\
    & \textstyle\bigcap_{k \in K} (N_k+i) = \textstyle\bigcap_{k \in K} N_k+i.
      \label{eqn:ts_cap}
  \end{align}
\end{proposition}
\begin{proof}
  All \eqref{eqn:ts_mon}--\eqref{eqn:ts_cap} follow directly
  from~\eqref{eqn:shft}.
\end{proof}
	
Based on~\eqref{eqn:ts_assoc}, we may omit parentheses and write
$M + j + i$ instead of $(M + i) + j$. Also, we write $M - i$ to denote
$M + (- i)$.
	
Attribute implications annotated by time points are formulas, i.e.,
syntactic notions for which we define their semantics (interpretation)
as follows.
	
\begin{definition}\label{def:sem}\upshape%
  A formula \emph{$A \Rightarrow B$ is true in $M \subseteq \alltatr$}
  whenever, for each $i \in \Z$,
  \begin{align}
    \text{if } A+i \subseteq M \text{, then } B+i \subseteq M
    \label{eqn:ifthen}
  \end{align}
  and we denote the fact by $M \models A \Rightarrow B$.
\end{definition}
	
\begin{remark}\label{rem:trivB}
  (a) The value of $i$ in the definition may be understood as a
  sliding time point. Moreover, $A+i$ and $B+i$ represent sets of
  attributes annotated by \emph{absolute time points} considering $i$
  as the current time point. Note that using~\eqref{eqn:shft}, the
  condition \eqref{eqn:ifthen} can be equivalently restated as
  ``$A \subseteq M - i$ implies $B \subseteq M - i$,'' i.e., instead
  of shifting the antecedents and consequents of the formula, we may
  shift the set $M$.
		
  (b) Observe that $A \Rightarrow B$ is trivially true in $M$ whenever
  $B \subseteq A$ because in that case \eqref{eqn:ifthen} trivially
  holds for any $i$. By definition, $A \Rightarrow B$ is not true in
  $M$, written $M \nmodels A \Rightarrow B$ if{}f there is $i$ such
  that $A+i \subseteq M$ and $B+i \nsubseteq M$.  In words, in the
  time point $i$, $M$ has all the attributes of $A$ but does not have
  an attribute in $B$, i.e., the time point $i$ serves as a
  counterexample.
\end{remark}
	
\begin{figure}[t]
  \centering%
  \setlength{\tabcolsep}{0.7ex}%
  \renewcommand{\arraystretch}{0.65}%
  \def\atr#1#2{\ensuremath{\mathtt{\lowercase{#1#2}}}}%
  \def\X{$\boldsymbol{\times}$} \def\O{}
  \begin{tabular}{|c|*{9}{|c}|}
    \hline
    &$\atr{R}{n}$&$\atr{R}{l}$&$\atr{R}{m}$
    &$\atr{T}{v}$&$\atr{T}{c}$&$\atr{T}{m}$
    &$\atr{W}{l}$&$\atr{W}{m}$&$\atr{W}{s}$ 
    \\
    \hline\hline
    15&\X&\O&\O&\O&\X&\O&\X&\O&\O\\
    16&\X&\O&\O&\O&\X&\O&\O&\X&\O\\
    17&\O&\X&\O&\X&\O&\O&\X&\O&\O\\
    18&\O&\O&\X&\O&\X&\O&\O&\X&\O\\
    19&\X&\O&\O&\O&\X&\O&\O&\O&\X\\
    20&\X&\O&\O&\O&\X&\O&\X&\O&\O\\
    21&\X&\O&\O&\O&\X&\O&\O&\X&\O\\
    22&\X&\O&\O&\O&\X&\O&\O&\X&\O\\
    23&\X&\O&\O&\O&\X&\O&\X&\O&\O\\
    24&\O&\X&\O&\O&\O&\X&\O&\X&\O\\
    25&\X&\O&\O&\O&\X&\O&\X&\O&\O\\
    26&\O&\O&\X&\O&\X&\O&\O&\X&\O\\
    27&\O&\X&\O&\O&\X&\O&\O&\X&\O\\
    28&\X&\O&\O&\O&\X&\O&\O&\X&\O\\
    29&\X&\O&\O&\O&\O&\X&\O&\X&\O\\
    \hline
  \end{tabular}
  \label{fig:weather}
  \caption{Daily weather observation from an airport station.}
\end{figure}
	
\begin{example}%
  \def\atr#1#2{\ensuremath{\mathtt{\lowercase{#1#2}}}}%
  One particular example of a subset $M$ of $\alltatr$ can be a daily
  weather observation from an airport station. For instance, we can
  consider $Y$ as
  \begin{align*}
    Y = \{
    \atr{R}{n},\atr{R}{l},\atr{R}{m},
    \atr{T}{v},\atr{T}{c},\atr{T}{m},
    \atr{W}{l},\atr{W}{m},\atr{W}{s}\},
  \end{align*}
  where the attributes have the following meaning: ``no rainfall''
  (denoted $\atr{R}{n}$), ``light rainfall'' (denoted $\atr{R}{l}$),
  ``moderate rainfall'' (denoted $\atr{R}{m}$),
  % [0 - 0.1), [0.1 - 5), [5 - )
  ``temperature is very cold'', (denoted $\atr{T}{v}$), ``temperature
  is cold'', (denoted $\atr{T}{c}$) ``temperate is mild'', (denoted
  $\atr{T}{m}$)
  % [0 - 10) [10 - 16) [16 - )
  ``light wind'' (denoted $\atr{W}{l}$), ``moderate wind'' (denoted
  $\atr{W}{m}$), and ``strong wind'' (denoted $\atr{W}{s}$).
  % [2 - 12) [12 - 35) [35 - )
  A subset of $\alltatr$ may be depicted as a two-dimensional table
  with rows corresponding to time points, columns corresponding to
  attributes in $Y$, and crosses and blanks in the table, indicating
  whether attributes annotated by time points belong to the
  subset. For instance, if $M$ is given by the table in
  Figure~\ref{fig:weather}, then $\atr{R}{n}^{15} \in M$,
  $\atr{R}{l}^{15} \not\in M$, etc\footnote{%
    The data is based on discretization of real meteorological
    information for Aug 14 which can be found at
    \url{http://www.bom.gov.au/climate/dwo/IDCJDW0100.shtml}.}.  In
  this case, we have
  $M \models \{\atr{W}{l}^0,\atr{W}{m}^1\} \Rightarrow
  \{\atr{T}{c}^3\}$.
  On the other hand,
  $M \nmodels \{\atr{W}{m}^0,\atr{W}{l}^1\} \Rightarrow
  \{\atr{T}{c}^3,\atr{R}{m}^{3},\atr{T}{c}^4\}$
  because for $i = 22$, we have
  $\{\atr{W}{m}^0,\atr{W}{l}^1\} + 22 =
  \{\atr{W}{m}^{22},\atr{W}{l}^{23}\}\subseteq M$
  and
  $\{\atr{T}{c}^3,\atr{R}{m}^{3},\atr{T}{c}^4\} + 22 =
  \{\atr{T}{c}^{25},\atr{R}{m}^{25},\atr{T}{c}^{26}\} \nsubseteq M$.
\end{example}

We consider the following notions of a theory and a model:
	
\begin{definition}\label{def:thmod}\upshape%
  Let $\Sigma$ be a set of formulas (called a \emph{theory}).  A
  subset $M \subseteq \alltatr$ is called a \emph{model of $\Sigma$}
  if $M \models A \Rightarrow B$ for all $A \Rightarrow B \in \Sigma$.
  The system of all models of $\Sigma$ is denoted by
  $\mathrm{Mod}(\Sigma)$, i.e.,
  \begin{align}
    \mathrm{Mod}(\Sigma) &=
    \bigl\{M \subseteq \alltatr \|
    M \models A \Rightarrow B
    \text{ for all } A \Rightarrow B \in \Sigma\bigr\}.
    \label{eqn:Mod}
  \end{align}
\end{definition}
	
In general, $\mathrm{Mod}(\Sigma)$ is infinite and there may be
theories that do not have any finite model. For instance, consider a
theory containing $\emptyset \Rightarrow \{y^0\}$.
	
We now turn our attention to the structure of systems of all models of
attribute implications annotated by time points. In case of the
ordinary attribute implications, it is well known that systems of
their models are exactly closure systems in $Y$. Interestingly, the
systems of models in our case are exactly the algebraic closure
systems which are closed under time shifts. This additional closure
property is introduced by the following definition.
	
\begin{definition}
  A system $\mathcal{S} \subseteq 2^{\alltatr}$ of subsets of
  $\alltatr$ is called \emph{closed under time shifts} whenever
  $M+i \in \mathcal{S}$ for all $M \in \mathcal{S}$ and $i \in \Z$.
\end{definition}
	
We first show that $\mathrm{Mod}(\Sigma)$ is a closure system closed
under time shifts:
	
\begin{theorem}\label{thm:mods}
  Let $\Sigma$ be a theory. Then, $\mathrm{Mod}(\Sigma)$ is closed
  under arbitrary intersections and time shifts.
\end{theorem}
\begin{proof}
  The fact that $\mathrm{Mod}(\Sigma)$ is closed under arbitrary
  intersections follows by analogous arguments as in the case of
  ordinary attribute implications taking into account
  that~\eqref{eqn:ifthen} must hold for all $i \in \Z$. That is, for
  any $\mathcal{M} \subseteq \mathrm{Mod}(\Sigma)$ and arbitrary
  $A \Rightarrow B \in \Sigma$, we reason as follows.  If
  $A + i \subseteq \textstyle\bigcap \mathcal{M}$, then
  $A + i \subseteq M$ for all $M \in \mathcal{M}$ and thus
  $B + i \subseteq M$ for all $M \in \mathcal{M}$ because
  $\mathcal{M} \subseteq \mathrm{Mod}(\Sigma)$. Therefore,
  $B + i \subseteq \textstyle\bigcap \mathcal{M}$, proving
  $\textstyle\bigcap \mathcal{M} \models A \Rightarrow B$ which
  further gives
  $\textstyle\bigcap \mathcal{M} \in \mathrm{Mod}(\Sigma)$ since
  $A \Rightarrow B \in \Sigma$ was arbitrary.
		
  In order to show that $\mathrm{Mod}(\Sigma)$ is closed under time
  shifts, take $M \in \mathrm{Mod}(\Sigma)$ and $j \in \Z$. It
  suffices to prove that $M+j \in \mathrm{Mod}(\Sigma)$. In order to
  see that, take $A \Rightarrow B \in \Sigma$. If $A+i \subseteq M+j$,
  then $A+(i-j) \subseteq M$ and thus $B+(i-j) \subseteq M$ because
  $M \in \mathrm{Mod}(\Sigma)$ and $A \Rightarrow B \in
  \Sigma$.
  Therefore, $B+i \subseteq M+j$, i.e.,
  $M + j \models A \Rightarrow B$ for arbitrary
  $A \Rightarrow B \in \Sigma$, showing
  $M+j \in \mathrm{Mod}(\Sigma)$.
\end{proof}
	
Taking into account Theorem~\ref{thm:mods}, for each theory $\Sigma$,
we may consider a closure operator induced by $\mathrm{Mod}(\Sigma)$
which maps each $M \subseteq \alltatr$ to the least model of $\Sigma$
containing $M$.
	
\begin{definition}
  Let $\Sigma$ be a theory. For each $M \subseteq \alltatr$, we put
  \begin{align}
    [M]_\Sigma &=
                 \textstyle\bigcap\{N \in \mathrm{Mod}(\Sigma) \| 
                 M \subseteq N\}
                 \label{eqn:semclos}
  \end{align}
  and call $[M]_\Sigma$ the \emph{semantic closure} of $M$ under
  $\Sigma$.
\end{definition}
	
Using the well-known relationship between closure operators and
closure systems, $[{\cdots}]_\Sigma$ defined by~\eqref{eqn:semclos} is
indeed a closure operator. Note that in general, $[M]_\Sigma$ is
infinite even if $Y$ and $M$ are finite. This is in contrast with the
ordinary attribute implications using finite $Y$. Nevertheless, in our
setting we can prove that even if $[M]_\Sigma$ is infinite, it can be
obtained as a union of finitely generated elements of
$\mathrm{Mod}(\Sigma)$, showing that $\mathrm{Mod}(\Sigma)$ is in fact
an \emph{algebraic} closure system.
	
\begin{theorem}\label{th:algmod}
  Let $\Sigma$ be a theory. For each $M \subseteq \alltatr$, we have
  \begin{align}
    [M]_{\Sigma} = \textstyle\bigcup
    \{[N]_{\Sigma}\| N \text{ is finite subset of } M\}.
    \label{eqn:algmod}
  \end{align}
\end{theorem}
\begin{proof}
  Observe that the monotony of $[{\cdots}]_\Sigma$ yields
  $[N]_{\Sigma} \subseteq [M]_{\Sigma}$ for any finite $N \subseteq M$
  and thus the ``$\supseteq$''-part of~\eqref{eqn:algmod} is obvious.
		
  For brevity, put
  $\mathcal{M} = \{[N]_{\Sigma}\| N \text{ is finite subset of } M\}$.
  In order to prove the ``$\subseteq$''-part of~\eqref{eqn:algmod}, it
  suffices to show that $\bigcup\mathcal{M}$ is a model of $\Sigma$
  which contains~$M$ because $[M]_\Sigma$ is the least model of
  $\Sigma$ containing $M$. For any $y^i \in M$, we have
  $[\{y^i\}]_\Sigma \in \mathcal{M}$ and thus
  $y^i \in [\{y^i\}]_\Sigma \subseteq \bigcup\mathcal{M}$ by the
  extensivity of $[{\cdots}]_\Sigma$ which proves
  $M \subseteq \bigcup\mathcal{M}$.
		
  Now, take any $A \Rightarrow B \in \Sigma$ and suppose that
  $A + i \subseteq \bigcup\mathcal{M}$. Observe that for every
  $y^j \in A+i$ there is $[N_{y^j}]_\Sigma \in \mathcal{M}$ such that
  $y^j \in [N_{y^j}]_\Sigma$. Moreover, the fact that $A+i$ is finite
  yields that $\bigcup\{N_{y^j} \,|\, y^j \in A+i\}$ is finite and we
  thus have
  $\bigl[\bigcup\{N_{y^j} \,|\, y^j \in A+i\}\bigr]_\Sigma \in
  \mathcal{M}$.
  Clearly,
  $A + i \subseteq \bigl[\bigcup\{N_{y^j} \,|\, y^j \in
  A+i\}\bigr]_\Sigma$
  and thus it follows that
  $B + i \subseteq \bigl[\bigcup\{N_{y^j} \,|\, y^j \in
  A+i\}\bigr]_\Sigma \subseteq \bigcup\mathcal{M}$
  because $A \Rightarrow B \in \Sigma$. Altogether,
  $\bigcup\mathcal{M} \models A \Rightarrow B$ and so
  $\bigcup\mathcal{M} \in \mathrm{Mod}(\Sigma)$.
\end{proof}
	
Using Theorem~\ref{th:algmod}, we may establish that each algebraic
closure system closed under time shifts is a system of models of some
theory consisting of attribute implications annotated by time points.
Before we go to the proof, we show how the property of being closed
under time shifts can be formulated in terms of closure operators.
	
\begin{lemma}\label{semclo-shift}
  Let $\mathcal{S}$ be a closure system which is closed under
  arbitrary time shifts and let $\mathrm{C}_\mathcal{S}$ be the
  induced closure operator.  For each $M \subseteq \alltatr$ and
  $i \in \Z$,
  \begin{align}
    \mathrm{C}_\mathcal{S}(M + i) = \mathrm{C}_\mathcal{S}(M) + i.
    \label{eqn:sh_invar}
  \end{align}
\end{lemma}
\begin{proof}
  ``$\subseteq$'': Since $\mathcal{S}$ is closed under time shifts, we
  get $\mathrm{C}_\mathcal{S}(M) + i \in \mathcal{S}$. In addition,
  $M + i \subseteq \mathrm{C}_\mathcal{S}(M) + i$ on account of the
  extensivity of $\mathrm{C}_\mathcal{S}$ and~\eqref{eqn:ts_mon}.
  Therefore,
  $\mathrm{C}_\mathcal{S}(M + i) \subseteq \mathrm{C}_\mathcal{S}(M) +
  i$ by monotony and idempotency of $\mathrm{C}_\mathcal{S}$.
		
  ``$\supseteq$'': The extensivity of $\mathrm{C}_\mathcal{S}$ gives
  $M + i \subseteq \mathrm{C}_\mathcal{S}(M+i)$ and thus
  $M \subseteq \mathrm{C}_\mathcal{S}(M+i)-i$.  Moreover,
  $\mathrm{C}_\mathcal{S}(M+i)-i \in \mathcal{S}$ because
  $\mathcal{S}$ is closed under time shifts and thus
  $\mathrm{C}_\mathcal{S}(M) \subseteq \mathrm{C}_\mathcal{S}(M+i)-i$
  which gives
  $\mathrm{C}_\mathcal{S}(M)+i \subseteq \mathrm{C}_\mathcal{S}(M+i)$.
\end{proof}
	
\begin{lemma}
  Let $\mathrm{C}$ be a closure operator satisfying
  $\mathrm{C}(M+i) = \mathrm{C}(M) + i$ for each
  $M \subseteq \alltatr$ and $i \in \Z$. Then, the system
  $\mathcal{S}_\mathrm{C}$ of all fixed points of $\mathrm{C}$ is
  closed under arbitrary time shifts.
\end{lemma}
\begin{proof}
  Take $M \in \mathcal{S}_\mathrm{C}$ and any $i \in \Z$, i.e.,
  $M \subseteq \alltatr$ such that $M = \mathrm{C}(M)$.  Clearly,
  $M + i = \mathrm{C}(M) + i$ and since
  $\mathrm{C}(M) + i = \mathrm{C}(M+i)$, we get
  $M + i = \mathrm{C}(M+i)$, proving that
  $M + i \in \mathcal{S}_\mathrm{C}$.
\end{proof}
	
The previous two lemmas give the following consequence.
	
\begin{corollary}
  A closure system $\mathcal{S}$ is closed under arbitrary time shifts
  if{}f the corresponding closure operator $\mathrm{C}_\mathcal{S}$
  satisfies~\eqref{eqn:sh_invar}.  \qed
\end{corollary}
	
Based on our previous observations, we may now establish the
connection between systems of models of attribute implications
annotated by time points and algebraic closure systems closed under
time shifts.
	
\begin{theorem}\label{thm:mods2}
  Let $\mathcal{S} \subseteq 2^{\alltatr}$ be an algebraic closure
  system which is closed under time shifts. Then, there is a theory
  $\Sigma$ such that $\mathcal{S} = \mathrm{Mod}(\Sigma)$.
\end{theorem}
\begin{proof}
  Assume that $\mathrm{C}_\mathcal{S}$ is the closure operator induced
  by $\mathcal{S}$ and put
  \begin{align*}
    \Sigma = \{A \Rightarrow B \,|\,
    A \subseteq \alltatr \text{, } B \subseteq 
    \mathrm{C}_\mathcal{S}(A) \text{, and } A,B \text{ are finite}\}.
  \end{align*}
  We show that $\mathcal{S} = \mathrm{Mod}(\Sigma)$ by proving that
  both inclusions hold.
		
  ``$\subseteq$'': Take $M \in \mathcal{S}$ and a finite
  $B \subseteq \mathrm{C}_\mathcal{S}(A)$ for a finite
  $A \subseteq \alltatr$. We now check that
  $M \models A \Rightarrow B$. Assume that $A + i \subseteq M$.  Then,
  $A \subseteq M - i$ and by the monotony of $\mathrm{C}_\mathcal{S}$
  and utilizing~\eqref{eqn:sh_invar}, we have
  $\mathrm{C}_\mathcal{S}(A) \subseteq \mathrm{C}_\mathcal{S}(M - i) =
  \mathrm{C}_\mathcal{S}(M) - i = M - i$
  which yields that $B \subseteq M - i$, i.e., $B + i \subseteq M$,
  showing $M \models A \Rightarrow B$. As a consequence,
  $\mathcal{S} \subseteq \mathrm{Mod}(\Sigma)$.
		
  ``$\supseteq$'': We let $M\in \mathrm{Mod}(\Sigma)$ and prove that
  $M \in \mathcal{S}$ which means to prove that
  $\mathrm{C}_\mathcal{S}(M) = M$. Since $\mathcal{S}$ is an algebraic
  closure system, it suffices to check that
  $\mathrm{C}_\mathcal{S}(A) \subseteq M$ for each finite
  $A \subseteq M$.  Assuming that $A \subseteq M$ and $A$ is finite,
  take any finite $B \subseteq \mathrm{C}_\mathcal{S}(A)$. By
  definition, $A \Rightarrow B \in \Sigma$ and since
  $M \in \mathrm{Mod}(\Sigma)$, we get that for $i = 0$,
  $A + 0 \subseteq M$ implies $B + 0 \subseteq M$.  Since $A + 0 = A$
  and $A \subseteq M$, we therefore obtain $B = B + 0 \subseteq
  M$.
  Since $B$ was an arbitrary finite subset of
  $\mathrm{C}_\mathcal{S}(A)$, we conclude that
  $\mathrm{C}_\mathcal{S}(A) \subseteq M$.
\end{proof}
	
We now define semantic entailment of formulas and explore its
properties.  The notion is defined the usual way using the notion of a
model introduced before.
	
\begin{definition}\label{def:entail}\upshape%
  Let $\Sigma$ be a theory. Formula \emph{$A \Rightarrow B$ is
    semantically entailed by $\Sigma$} if $M \models A \Rightarrow B$
  for each $M \in \mathrm{Mod}(\Sigma)$.
\end{definition}

The following lemma justifies the description of time points in
attribute implications as \emph{relative} time points.  Namely, it
states that each $A \Rightarrow B$ semantically entails all formulas
resulting by shifting the antecedent and consequent of
$A \Rightarrow B$ by a constant factor.
	
\begin{lemma}\label{tai-relativ}
  $\{A \Rightarrow B\} \models A+i \Rightarrow B+i$.
\end{lemma}
\begin{proof}
  Take $M \in \mathrm{Mod}(\{A \Rightarrow B\})$ and let
  $(A+i)+j \subseteq M$. Then, $A+i \subseteq M-j$ and by
  Theorem~\ref{thm:mods}, we get
  $M-j \in \mathrm{Mod}(\{A \Rightarrow B\})$ which yields
  $B+i \subseteq M-j$ and thus $(B+i)+j \subseteq M$, proving
  $M \models A+i \Rightarrow B+i$
\end{proof}

Analogously as for the classic attribute implications, the
semantic entailment of $A \Rightarrow B$ by a theory $\Sigma$
can be checked using the least model of $\Sigma$ generated by
$A$ as it is shown in the following theorem.

\begin{theorem}\label{th:semclos}
  For any $\Sigma$ and $A \Rightarrow B$, the following
  conditions are equivalent: \bgroup%
  \setlength{\leftmargini}{3em}%
  \begin{enumerate}\parskip=\smallskipamount%
  \item[\itm{1}] $\Sigma \models A \Rightarrow B$,
  \item[\itm{2}] $[A]_\Sigma \models A \Rightarrow B$,
  \item[\itm{3}] $B \subseteq [A]_\Sigma$.
  \end{enumerate}%
  \egroup%
\end{theorem}
\begin{proof}
  Clearly, \itm{1} implies \itm{2} since
  $[A]_\Sigma \in \mathrm{Mod}(\Sigma)$; \itm{2} implies
  \itm{3} because $A+0 \subseteq [A]_\Sigma$. Assume that
  \itm{3} holds and take $M \in \mathrm{Mod}(\Sigma)$ and
  $i \in \Z$ such that $A+i \subseteq M$. Then,
  $A \subseteq M-i$ and thus
  $B \subseteq [A]_\Sigma \subseteq [M-i]_\Sigma =
  [M]_\Sigma-i$
  by~\eqref{eqn:sh_invar} from which it follows that
  $B+i \subseteq [M]_\Sigma = M$, proving \itm{1}.
\end{proof}

We conclude this section by notes on the propositional
semantics of our formulas. The classic attribute implications
on finite $Y$ can be understood as propositional
formulas. Namely, an attribute implication of the
from~\eqref{eqn:impl} can be seen as a propositional formula
\begin{align}
  \bigl(y_1 \logand \cdots \logand y_m\bigr)
  \Rightarrow
  \bigl(z_1 \logand \cdots \logand z_n\bigr),
  \label{eqn:prop_impl}
\end{align}
where $\logand$ is the symbol for conjunction and
$y_1,\ldots,y_m,z_1,\ldots,z_n$ are propositional variables.
Thus, \eqref{eqn:prop_impl} may be called a propositional
counterpart of~\eqref{eqn:impl}. Obviously, there are in
general several propositional counterparts of~\eqref{eqn:impl}
since formulas equivalent to~\eqref{eqn:prop_impl} in sense of
the propositional logic result, e.g., by rearranging the
propositional variables $y_1,\ldots,y_m,z_1,\ldots,z_n$ in a
different order. We neglect this aspect and always consider a
fixed propositional counterpart of each attribute implication.
It can be shown that if one takes the propositional
counterparts of attribute implications, then their semantic
entailment in sense of the propositional logic coincides with
the semantic entailment as it is defined for attribute
implications. We now show that an analogous correspondence can
also be established in our case.

We start by considering the following notation.  For any
finite $A,B \subseteq \alltatr$ and for any
$M \subseteq \alltatr$, we put $M \models_\PL A \Rightarrow B$
whenever $A \nsubseteq M$ or $B \subseteq M$. That is,
$M \models_\PL A \Rightarrow B$ means that $A \Rightarrow B$
is true in $M$ as a classical attribute implication. Clearly,
$M \models_\PL A \Rightarrow B$ does not imply that
$M \models A \Rightarrow B$ in sense of
Definition~\ref{def:sem}. Moreover, we may introduce the set
of models of $\Sigma$ in the classic sense:
\begin{align}
  \mathrm{Mod}^\PL(\Sigma) &=
                             \bigl\{M \subseteq \alltatr \|
                             M \models_\PL A \Rightarrow B
                             \text{ for all } A \Rightarrow B \in \Sigma\bigr\}
                             \label{eqn:MonPL}
\end{align}
and put $\Sigma \models_\PL A \Rightarrow B$ whenever
$M \models_\PL A \Rightarrow B$ for all
$M \in \mathrm{Mod}^\PL(\Sigma)$.  Therefore, $\models_\PL$
denotes the semantic entailment of attribute implications in
the classic sense. Again, $\models_\PL$ is in general
different from $\models$ introduced in
Definition~\ref{def:entail} but we can establish the following
characterization:

\begin{theorem}\label{thm:PL_sem}
  Let $\Sigma$ be a theory and let
  \begin{align}
    \Sigma^\PL &= \{A+i \Rightarrow B+i \,|\,
                 A \Rightarrow B \in \Sigma \text{ and } i \in \Z\}.
                 \label{eqn:SigmaPL}
  \end{align}
  Then $\mathrm{Mod}(\Sigma) = \mathrm{Mod}^\PL(\Sigma^\PL)$.
  As a consequence, for each $A \Rightarrow B$, we have
  $\Sigma \models A \Rightarrow B$ if{}f\/
  $\Sigma^\PL \models_\PL A \Rightarrow B$.
\end{theorem}
\begin{proof}
  The first part of the claim is easy to see. Indeed, for each
  $A \Rightarrow B$ we have
  $M \in \mathrm{Mod}(\{A \Rightarrow B\})$ if{}f for each
  $i \in \Z$, we have $A+i \subseteq M$ implies
  $B+i \subseteq M$ which is true if{}f
  $M \in \mathrm{Mod}^\PL(\{A+i \Rightarrow B+i \,|\, i \in
  \Z\})$.
  Hence, it follows that
  $\mathrm{Mod}(\Sigma) = \mathrm{Mod}^\PL(\Sigma^\PL)$.
  
  Now, assume that $\Sigma \models A \Rightarrow B$ and take
  $M \in \mathrm{Mod}^\PL(\Sigma^\PL)$ such that
  $A \subseteq M$.  Then $A + 0 \subseteq M$ and
  $M \in \mathrm{Mod}(\Sigma)$ and thus
  $B = B + 0 \subseteq M$, proving that
  $\Sigma^\PL \models_\PL A \Rightarrow B$.  Conversely, let
  $\Sigma^\PL \models_\PL A \Rightarrow B$ and
  $A+i \subseteq M$ for $M \in \mathrm{Mod}(\Sigma)$. That is,
  we have $A \subseteq M - i$ and, owing to
  Theorem~\ref{thm:mods},
  $M - i \in \mathrm{Mod}(\Sigma) =
  \mathrm{Mod}^\PL(\Sigma^\PL)$.
  As a consequence of $M - i \models_\PL A \Rightarrow B$, we
  get $B \subseteq M - i$ and thus $B + i \subseteq M$,
  showing $\Sigma \models A \Rightarrow B$.  Altogether,
  $\Sigma \models A \Rightarrow B$ if{}f\/
  $\Sigma^\PL \models_\PL A \Rightarrow B$.
\end{proof}

Now, based on Theorem~\ref{thm:PL_sem}, we may argue that for
each $\Sigma$ there is a set of propositional formulas
$\Sigma'$ such that the propositional counterpart of
$A \Rightarrow B$ follows by $\Sigma'$ in sense of the
propositional logic. Indeed, $\Sigma'$ can be taken as the set
of propositional counterparts to all formulas in $\Sigma^\PL$:
Owing to Theorem~\ref{thm:PL_sem}, $A \Rightarrow B$ follows
by $\Sigma^\PL$ as a classic attribute implication over (a
denumerable set of attributes) $\alltatr$ and thus the
propositional counterpart of $A \Rightarrow B$ follows by the
propositional counterparts to all formulas in $\Sigma^\PL$.

%%%%%%%%%%%%%%%%%%%%%%%%%%%%%%%%%%%%%%%%%%%%%%%%%%%%%%%%%%%%%%%%%%%%%%%%%%%%%%%%
%%%%% DEDUCTION SYSTEMS AND COMPLETE AXIOMATIZATIONS
%%%%%%%%%%%%%%%%%%%%%%%%%%%%%%%%%%%%%%%%%%%%%%%%%%%%%%%%%%%%%%%%%%%%%%%%%%%%%%%%
\section{Deduction Systems and Complete Axiomatizations}\label{sec:compl}
In this section, we present a deduction system for our formulas and a
related notion of provability which represents the syntactic entailment of
formulas.
The provability is based on an extension of the Armstrong axiomatic
system~\cite{Arm:Dsdbr} which is well known mainly in database
systems~\cite{Mai:TRD}. The extension we propose accommodates the fact
that time points in formulas are relative. The deductive system we use
consists of the following deduction rules.

\begin{definition}\label{def:ded}
  We introduce the following \emph{deduction rules}:
  \bgroup%
  \setlength{\leftmargini}{3.5em}%
  \begin{itemize}\parskip=\smallskipamount%
  \item[(Ax)]
    infer $A {\cup} B \Rightarrow A$,
  \item[(Cut)]
    from $A \Rightarrow B$ and
    $B{\cup}C \Rightarrow D$ infer $A{\cup}C \Rightarrow D$,
  \item[(Shf)]
    from $A \Rightarrow B$ infer
    $A+i \Rightarrow B+i$,
  \end{itemize}
  \egroup%
  \noindent%
  where $i \in \Z$ and $A,B,C,D$ are arbitrary finite subsets of $\alltatr$.
\end{definition}

\begin{remark}
  (a)
  Note that there are several equivalent systems which are called
  the Armstrong systems~\cite{Mai:TRD}.
  In our presentation, the rule (Ax) can be seen as a nullary deduction
  rule which is an axiom scheme, i.e., each $A {\cup} B \Rightarrow A$
  may be called an axiom.
  (Cut) and (Shf) are binary and unary deduction rules, respectively.
  In the classic case, (Ax) and (Cut) form a system which is equivalent
  to that from~\cite{Arm:Dsdbr}. We call the additional rule (Shf)
  the rule of ``time shifts.'' Also note that in the database
  literature, (Cut) is also referred to as the rule
  of pseudo-transitivity~\cite{Mai:TRD}.
  
  (b)
  The rules in Definition~\ref{def:ded} can be written as fractions with
  hypotheses (formulas preceding ``infer'') above the conclusion
  (formula following ``infer'') as
  \begin{align*}
    &\AX{A{\cup}B\Rightarrow A}, &
    &\CUT{A \Rightarrow B}{B{\cup}C \Rightarrow D}{A{\cup}C \Rightarrow D}, &
    &\SHF{A \Rightarrow B}{A+i \Rightarrow B+i}.
  \end{align*}
\end{remark}

Although we are going to use (Ax), (Cut), and (Shf) as the basic deduction
rules in our system, we define the notion of provability relatively to
a collection of deduction rules because we later investigate systems
consisting of other rules. Thus, a general \emph{deduction system} is
a set $\mathcal{R}$ of $n$-ary rules of the form
``from $\varphi_1,\ldots,\varphi_n$, infer $\psi$''.

\begin{definition}
  Let $\mathcal{R}$ be a deduction system.
  An $\mathcal{R}$-proof of $A \Rightarrow B$ by $\Sigma$ is
  a finite sequence $\delta_1,\ldots,\delta_n$ such that $\delta_n$
  equals $A \Rightarrow B$ and for each $i=1,\ldots,n$ we have
  \begin{itemize}\parskip=0pt%
  \item[\itm{1}]
    $\delta_i \in \Sigma$, or
  \item[\itm{2}]
    $\mathcal{R}$ contains a rule
    ``from $\varphi_1,\ldots,\varphi_n$ infer $\psi$''
    such that $\psi$ is equal to $\delta_i$ and
    we have $\{\varphi_1,\ldots,\varphi_n\} \subseteq \{\delta_j \,|\, j < i\}$.
  \end{itemize}
  We say that \emph{$A \Rightarrow B$ is $\mathcal{R}$-provable by $\Sigma$,}
  denoted $\Sigma \proves_\mathcal{R} A \Rightarrow B$, if there is
  an $\mathcal{R}$-proof of $A \Rightarrow B$ by $\Sigma$.
\end{definition}

If $\mathcal{R}$ consists solely of (Ax), (Cut), and (Shf), we write just
$\Sigma \proves A \Rightarrow B$ and call $A \Rightarrow B$ provable
by $\Sigma$. Analogously, we use the term ``proof'' instead of 
``$\mathcal{R}$-proof''. In the paper, we use the following
properties of provability.

\begin{proposition}\label{prop:deriv}
  For every finite $A,B,C,D \subseteq \alltatr$, we have
  \bgroup%
  \setlength{\leftmargini}{3.5em}%
  \begin{itemize}\parskip=0pt%
  \item[\Ref]
    $\proves A \Rightarrow A$,
  \item[\Wea]
    $\{A \Rightarrow C\} \proves A{\cup}B \Rightarrow C$,
  \item[\Acc]
    $\{A \Rightarrow B{\cup}C, C \Rightarrow D{\cup}E\}
    \proves A \Rightarrow B{\cup}C{\cup}D$,
  \item[\Add]
    $\{A \Rightarrow B, A \Rightarrow C\} \proves A \Rightarrow B{\cup}C$,
  \item[\Aug]
    $\{B \Rightarrow C\} \proves A{\cup}B \Rightarrow A{\cup}C$,
  \item[\Pro]
    $\{A \Rightarrow B{\cup}C\} \proves A \Rightarrow B$,
  \item[\Tra]
    $\{A \Rightarrow B,B \Rightarrow C\} \proves A \Rightarrow C$.
  \end{itemize}
  \egroup%
\end{proposition}
\begin{proof}
  The laws hold because our system is an extension of the Armstrong system
  in which the laws hold as well, see~\cite{Arm:Dsdbr,Mai:TRD}.
\end{proof}

Our inference system is sound in the usual sense:

\begin{theorem}[soundness]\label{thm:soundness}%
  If $\Sigma \proves A \Rightarrow B$ then
  $\Sigma \models A \Rightarrow B$.
\end{theorem}
\begin{proof}
  The proof goes by induction on the length of a proof,
  considering the facts that each axiom is true in all models,
  (Cut) is a sound deduction rule~\cite{Mai:TRD},
  and (Shf) is sound on account of Lemma~\ref{tai-relativ}.
  In a more detail, let $\delta_1,\ldots,\delta_n$ be a proof
  by $\Sigma$ and let $\Sigma \models \delta_i$ for all $i < j$.
  Then, if $\delta_j$ results by $\delta_i$ using $\Shf$ for some $i < j$,
  then $\Sigma \models \delta_i$ yields that
  $M \models \delta_i$ for all $M \in \mathrm{Mod}(\Sigma)$ and thus,
  using Lemma~\ref{tai-relativ}, $M \models \delta_j$ for all 
  $M \in \mathrm{Mod}(\Sigma)$, showing $\Sigma \models \delta_j$.
  The rest follows as in the classic case.
\end{proof}

In the proof of completeness, we utilize the notion of
a syntactic closure which is introduced as follows.

\begin{definition}
  Let $\Sigma$ be a theory. For each $M \subseteq \alltatr$, we put
  \begin{align}
    M_\Sigma^0 &= M,
    \\
    M_\Sigma^{n+1} &=  M_\Sigma^{n} \cup
                     \textstyle\bigcup\bigl\{F+i \|
                     E \Rightarrow F \in \Sigma \text{ and } 
                     E+i \subseteq M_\Sigma^{n}\bigr\},
                     \label{eqn:A_n+1}
    \\
    M_\Sigma^{\omega} &= \textstyle\bigcup_{n=0}^\infty M_\Sigma^{n}.
			\label{eqn:Aomega}
  \end{align}
  and call $M_\Sigma^{\omega}$
  the \emph{syntactic closure} of $M$ under $\Sigma$.
\end{definition}

By the Tarski fixpoint theorem~\cite{Ta:Altfta},
the operator which maps $M$ to $M_\Sigma^{\omega}$ defined
by~\eqref{eqn:Aomega} is indeed a closure operator,
so the term ``closure'' in the name syntactic closure is
appropriate. The following observation shows that the term
``syntactic'' is also appropriate since closures are directly
related to provability.

\begin{lemma}\label{le:synclos}
  Let $A,B \subseteq \alltatr$ be finite.
  Then, $B \subseteq A_\Sigma^{\omega}$ if{}f\/
  $\Sigma \proves A \Rightarrow B$.
\end{lemma}
\begin{proof}
  Suppose that $B \subseteq A_\Sigma^{\omega}$.
  Since $B$ is finite, there is $m$ such that $B \subseteq A_\Sigma^m$.
  Thus, in order to show that $\Sigma \proves A \Rightarrow B$,
  it suffices to check that for every $n$ and every
  finite $D \subseteq A_\Sigma^{n}$, we have $\Sigma \proves A \Rightarrow D$
  since then the claim readily follows for $D = B$ and $n = m$.
  By induction, assume the claim holds for $n$ and all finite
  $D \subseteq A_\Sigma^{n}$. Consider $n+1$ and take
  a finite $D \subseteq A_\Sigma^{n+1}$.
  Now, consider a finite
  \begin{align*}
    D' = \{\langle E \Rightarrow F,i\rangle \,|\,
    E \Rightarrow F \in \Sigma \text{ and } E+i \subseteq A_\Sigma^{n}\}
  \end{align*}
  such that
  \begin{align*}
    D \subseteq
    A_\Sigma^{n} \cup
    \textstyle\bigcup \{F+i \,|\, \langle E \Rightarrow F, i\rangle \in D'\}
    \subseteq A_\Sigma^{n+1}.
  \end{align*}
  Notice that since we assume $D$ finite, such finite $D'$ always exists.
  Now, by induction hypothesis,
  for each $\langle E \Rightarrow F, i\rangle \in D'$,
  we have $\Sigma \proves A \Rightarrow E+i$ owing to
  $E+i \subseteq A_\Sigma^{n} \subseteq A_\Sigma^{\omega}$.
  Furthermore, for $E \Rightarrow F \in \Sigma$, we have
  $\Sigma \proves E+i \Rightarrow F+i$ using (Shf).
  Thus, (Tra) gives $\Sigma \proves A \Rightarrow F+i$
  for each $\langle E \Rightarrow F, i\rangle \in D'$.
  In addition to that,
  $D \cap A_\Sigma^n \subseteq A_\Sigma^n$ and thus
  $\Sigma \proves A \Rightarrow D \cap A_\Sigma^n$.
  Since $D'$ is finite and
  $D \subseteq (D \cap A_\Sigma^n) \cup \bigcup
  \{F+i \,|\, \langle E \Rightarrow F, i\rangle \in D'\}$,
  $\Sigma \proves A \Rightarrow D$ follows by finitely many
  applications of $\Add$ and $\Pro$.
  As a consequence, $\Sigma \proves A \Rightarrow B$.
  
  Conversely, assume that $\Sigma \proves A \Rightarrow B$.
  By Theorem~\ref{thm:soundness}, $\Sigma \models A \Rightarrow B$.
  We show that $A_\Sigma^{\omega} \in \mathrm{Mod}(\Sigma)$.
  Take $E \Rightarrow F \in \Sigma$, $i \in \Z$ and
  let $E+i \subseteq A_\Sigma^{\omega}$. Since $E+i$ is finite,
  there must be $n$ such that $E+i \subseteq A_\Sigma^{n}$ and thus
  $F+i \subseteq A_\Sigma^{n+1} \subseteq A_\Sigma^{\omega}$, proving
  that $A_\Sigma^{\omega} \in \mathrm{Mod}(\Sigma)$.
  Now, $\Sigma \models A \Rightarrow B$ and
  $A + 0 = A \subseteq A_\Sigma^{\omega}$
  yields that $B+0 = B \subseteq A_\Sigma^{\omega}$.
\end{proof}

Note that Lemma~\ref{le:synclos} is in fact a syntactic counterpart
of Theorem~\ref{th:semclos}. Now, using previous observations,
we derive that our logic is complete:

\begin{theorem}[completeness]\label{th:scmpl}
  $\Sigma \proves A \Rightarrow B$ if{}f\/
  $\Sigma \models A \Rightarrow B$.
\end{theorem}
\begin{proof}
  If $\Sigma \nproves A \Rightarrow B$, we prove that
  there is $M \in \mathrm{Mod}(\Sigma)$
  such that $M \nmodels A \Rightarrow B$.
  Indeed, we show that one can take $A_\Sigma^{\omega}$ for $M$.
  By Lemma~\ref{le:synclos}, $\Sigma \nproves A \Rightarrow B$
  yields $B \nsubseteq A_\Sigma^{\omega}$.
  So, for $i = 0$, we have that
  $A+i = A \subseteq A_\Sigma^{\omega}$ and
  $B+i = B \nsubseteq A_\Sigma^{\omega}$,
  i.e., $A_\Sigma^{\omega} \nmodels A \Rightarrow B$.
  In addition to that, if $E + i \subseteq A_\Sigma^{\omega}$
  for $E \Rightarrow F \in \Sigma$ and $i \in \Z$,
  then $\Sigma \vdash A \Rightarrow E + i$ by Lemma~\ref{le:synclos}
  and so $\Sigma \vdash A \Rightarrow F + i$ using $\Shf$ and $\Tra$.
  Using Lemma~\ref{le:synclos} again, $F + i \subseteq A_\Sigma^{\omega}$
  which proves $A_\Sigma^{\omega} \in \mathrm{Mod}(\Sigma)$.
  The rest is a consequence of Theorem~\ref{thm:soundness}.
\end{proof}

As a corollary of the previous observations, we get the following assertion
showing that both the syntactic and semantic closures coincide.

\begin{theorem}\label{th:clos}
  For every $M \subseteq \alltatr$, we have $[M]_\Sigma = M_\Sigma^{\omega}$.
\end{theorem}
\begin{proof}
  We get $[M]_\Sigma \subseteq M_\Sigma^{\omega}$ since $[M]_\Sigma$ is the
  least model of $\Sigma$ containing $M$. Conversely,
  observe that for any $N \in \mathrm{Mod}(\Sigma)$ such that $M \subseteq N$,
  it follows that $M_\Sigma^{\omega} \subseteq N_\Sigma^{\omega} = N$.
  Hence, for $N$ being $[M]_\Sigma$,
  we get $M_\Sigma^{\omega} \subseteq [M]_\Sigma$.
\end{proof}

\begin{remark}\label{rem:shf_strong}
  Let us stress that the notions of semantic and syntactic entailment we
  have considered in our paper are different from their classic counterparts.
  Indeed, each attribute implication annotated by time points can also be
  seen as a classic attribute implication \emph{per se} because the sets
  $A$ and $B$ in $A \Rightarrow B$ are subsets of~$\alltatr$. Therefore,
  in addition to the semantic entailment from Definition~\ref{def:entail},
  we may consider the ordinary one which disregards the special role of
  time points. The same applies to the provability---the classic notion
  is obtained by omitting the rule (Shf). For instance,
  $\Sigma = \{\{x^1\} \Rightarrow \{y^2\},\{y^5\} \Rightarrow \{z^2\}\}$
  proves $\{x^4\} \Rightarrow \{y^5\}$ by (Shf) and thus 
  $\{x^4\} \Rightarrow \{z^2\}$ by (Tra). On the other hand, $\Sigma$
  does not prove $\{x^4\} \Rightarrow \{z^2\}$ without (Shf).
\end{remark}

\begin{remark}
  (a)
  We can show that our system of deduction rules consisting of
  $\Ax$, $\Cut$, and $\Shf$ is non-redundant, i.e., all the rules in the
  system are independent. Indeed, no formulas are provable by 
  $\Sigma = \emptyset$ using only $\Cut$ and $\Shf$ and thus $\Ax$
  is independent. Moreover, $\Cut$ is independent since all formulas
  provable by $\Sigma = \emptyset$ using only $\Ax$ and $\Shf$ are
  exactly all instances of $\Ax$. The independence of $\Shf$ follows
  by Remark~\ref{rem:shf_strong}.
  
  (b)
  Let us note that the deductive system in Definition~\ref{def:ded} is not
  minimal in terms of the number of deduction rules. Indeed, we may
  replace $\Cut$ and $\Shf$ by a single deduction rule
  \begin{align}
    \iCUT{A \Rightarrow B+i}{B{\cup}C \Rightarrow D}{
    A \cup (C+i) \Rightarrow D+i}.
  \end{align}
  Indeed, observe that $\Cut$ is a particular case of $\iCut$ for $i = 0$
  and $\Shf$ results by $\iCut$ and $\Ax$ for $A = B = \emptyset$. Conversely,
  $\{A \Rightarrow B+i, B{\cup}C \Rightarrow D\} \proves
  A {\cup} (C+i) \Rightarrow D+i$ because using~\eqref{eqn:ts_cup},
  the sequence
  \begin{align*}
    A \Rightarrow B+i,
    B{\cup}C \Rightarrow D,
    (B+i){\cup}(C+i) \Rightarrow D+i,
    A {\cup} (C+i) \Rightarrow D+i
  \end{align*}
  is a proof of $A {\cup} (C+i) \Rightarrow D+i$ using $\Cut$ and $\Shf$.
  As a consequence, the system consisting of $\Ax$, $\Cut$, and $\Shf$ is
  equivalent to $\Ax$ and $\iCut$.
  
  (c)
  An alternative deduction system for our logic can be based on $\Ref$
  instead of $\Ax$ and a single rule which is a modification of
  a simplification deduction rule~\cite{SL}.
  First, it is easily seen that $\Ax$ and $\Cut$ may be equivalently
  replaced by the following rule and $\Ref$:
  \begin{align}
    \SIM{A \Rightarrow B}{C \Rightarrow D}{
    A \cup (C \setminus B) \Rightarrow D}.
  \end{align}
  Indeed, $\Sim$ is a rule derivable by $\Ax$ and $\Cut$
  because the sequence
  \begin{align*}
    A \Rightarrow B,
    B{\cup}C \Rightarrow C,
    C \Rightarrow D,
    B{\cup}C \Rightarrow D,
    A{\cup}(C{\setminus}B) \Rightarrow D,
  \end{align*}
  is a proof of $A{\cup}(C{\setminus}B) \Rightarrow D$ by 
  $\{A \Rightarrow B,C \Rightarrow D\}$ using $\Ax$ and $\Cut$;
  apply the rule twice and observe that $B \cup C = B \cup (C \setminus B)$.
  Conversely, observe first that $\Ax$ is derivable by~$\Ref$ and $\Sim$
  because from $B \Rightarrow B$ and $A \Rightarrow A$ it follows
  that $B {\cup} (A {\setminus} B) \Rightarrow A$ that is,
  $A {\cup} B \Rightarrow A$. Moreover, $\Cut$ is derivable by
  $\Ref$ and $\Sim$ because the following sequence
  \begin{align*}
    C \Rightarrow C,
    \emptyset \Rightarrow \emptyset,
    C \Rightarrow \emptyset,
    A \Rightarrow B,
    B{\cup}C \Rightarrow D,
    A{\cup}((B{\cup}C){\setminus}B) \Rightarrow D,
    A{\cup}C \Rightarrow D,
  \end{align*}
  is a proof of $A{\cup}C \Rightarrow D$ by
  $\{A \Rightarrow B,B{\cup}C \Rightarrow D\}$ in which we have used
  $\Sim$ three times and utilized the fact that
  $C \cup ((A \cup ((B \cup C) \setminus B)) \setminus \emptyset) =
  A \cup C$. Altogether, $\Ax$ and $\Cut$ can indeed be replaced by
  $\Ref$ and $\Sim$. Note that $\Sim$ may be perceived even more natural
  than $\Cut$ because it is applicable to any two input formulas.
  Note that a rule analogous to $\Sim$ with the inferred formula being 
  $A \cup (C \setminus B) \Rightarrow B \cup D$ was first proposed
  by Darwen~\cite[page 140]{DaDa:RDW89_91}.
  Now, we may consider an extension of $\Sim$ which involves time shifts:
  \begin{align}
    \iSIM{A \Rightarrow B+i}{C \Rightarrow D}{
    A \cup ((C \setminus B) +i) \Rightarrow D+i}.
  \end{align}
  Analogously as in the case of $\iCut$, $\Sim$ is a particular case of
  $\iSim$ for $i = 0$ and $\Shf$ results by $\iSim$ and $\Ref$ for
  $A = B = \emptyset$.
  Therefore, the deductive system in Definition~\ref{def:ded} can be replaced
  by $\Ref$ and $\iSim$.
\end{remark}

We now focus on the order in which the deduction rules may be applied in
proofs. We show that each proof may be transformed into a normalized proof
which involves applications of deduction rules in a special order. First,
we show that $\Shf$ commutes with the other rules. Formally, we introduce
the property for a general deduction rule $R$ as follows:

Let $R$ be a deduction rule of the form ``from $\varphi_1,\ldots,\varphi_n$
infer $\psi$''. We say that $\Shf$ \emph{commutes} with $R$ if for any
formula $\chi$ which results by $\psi$ using $\Shf$ there are
$\varphi'_1,\ldots,\varphi'_n$ which result by $\varphi_1,\ldots,\varphi_n$
using \Shf, respectively, such that $\chi$ is provable
by $\{\varphi'_1,\ldots,\varphi'_n\}$ using $R$.

\begin{lemma}\label{le:commut}
  $\Shf$ commutes with $\Ax$, $\Cut$, and $\Shf$.
\end{lemma}
\begin{proof}
  Clearly, $\Shf$ commutes with $\Ax$ because the result of application
  of $\Shf$ to an instance of $\Ax$ is again an instance of $\Ax$. Moreover,
  $\Shf$ commutes with itself since $(A + i) + j$ equals $A + (i + j)$ for
  any $A \subseteq \alltatr$ and $i,j \in \Z$. Therefore, it remains to check
  that $\Shf$ commutes with $\Cut$. Consider formulas $A \Rightarrow B$
  and $B {\cup} C \Rightarrow D$ and the formula $A {\cup} C \Rightarrow D$
  which results by $\Cut$ and formula $(A {\cup} C)+i \Rightarrow D+i$ which
  results by $\Shf$. Clearly, if we apply $\Shf$ to $A \Rightarrow B$
  and $B {\cup} C \Rightarrow D$ for $i$, we obtain
  $A+i \Rightarrow B+i$ and $(B {\cup} C)+i \Rightarrow D+i$, respectively.
  The second formula equals $(B+i) {\cup} (C+i) \Rightarrow D+i$ and
  thus we may apply $\Cut$ to obtain $(A+i) \cup (C+i) \Rightarrow D+i$
  which equals $(A {\cup} C)+i \Rightarrow D+i$, proving that
  $\Shf$ commutes with $\Cut$.
\end{proof}

\begin{theorem}\label{th:commut}
  $\Sigma \proves A \Rightarrow B$ if{}f\/
  there is a finite $\Sigma' \subseteq \Sigma^\PL$
  such that 
  $\Sigma' \proves_\mathcal{R} A \Rightarrow B$ for
  $\mathcal{R}$ containing $\Ax$ and $\Cut$.
\end{theorem}
\begin{proof}
  In order to see the only-if part, assume that $\Sigma \proves A \Rightarrow B$
  which means there is a proof of $A \Rightarrow B$ by $\Sigma$. The proofs
  contains only finitely many formulas in $\Sigma$ and thus, we may consider
  a finite $\Sigma'' \subseteq \Sigma$ such that
  $\Sigma'' \proves A \Rightarrow B$. Moreover, the proof contains only
  finitely many applications of $\Shf$ and, using Lemma~\ref{le:commut},
  there is a proof of $A \Rightarrow B$ by $\Sigma''$ which starts by formulas
  in $\Sigma''$, then continues with applications of $\Shf$, and terminates
  with formulas derived without using $\Shf$. Therefore, there is
  a finite $\Sigma' \subseteq (\Sigma'')^\PL \subseteq \Sigma^\PL$
  such that $A \Rightarrow B$ is provable by $\Sigma'$ using only
  $\Ax$ and $\Cut$. The if-part of the assertion is easy to see.
\end{proof}

The previous observation allows us to introduce special derivation sequences
which represent proofs in a normalized form in that all utilized deduction rules
are applied in a particular order. The proofs are constructed using deduction
rules $\Ref$, $\Shf$, $\Acc$, and $\Pro$, see Proposition~\ref{prop:deriv}.

\begin{definition}\label{def:normseq}
  A finite sequence of formulas $\varphi_1,\ldots,\varphi_n$
  is called a \emph{normalized derivation sequence} of $A \Rightarrow B$
  using formulas in $\Sigma$ if the sequence
  \begin{enumerate}\parskip=0pt
  \item[\itm{1}]
    starts with finitely many formulas in $\Sigma$;
  \item[\itm{2}]
    continues by formulas obtained using $\Shf$ applied to formulas in \itm{1};
  \item[\itm{3}]
    continues by $A \Rightarrow A$;
  \item[\itm{4}]
    continues by formulas obtained using $\Acc$ whose first argument is
    the preceding formula and the second argument is a formula in
    \itm{1} or \itm{2};
  \item[\itm{5}]
    terminates with $A \Rightarrow B$
    which results by the 
    preceding formula by $\Pro$.
  \end{enumerate}
\end{definition}

Normalized derivation sequences are sufficient and adequate means for
determining provability of formulas:

\begin{theorem}\label{th:norm}
  $\Sigma \proves A \Rightarrow B$ if{}f\/
  there is a normalized derivation sequence of $A \Rightarrow B$
  using formulas in $\Sigma$.
\end{theorem}
\begin{proof}
  The if-part follows directly by the fact that a normalized 
  derivation sequence of $A \Rightarrow B$ using formulas in $\Sigma$ is
  a proof of $A \Rightarrow B$ by $\Sigma$
  using $\Ref$, $\Shf$, $\Acc$, and $\Pro$. Since all of them are rules
  derivable by $\Ax$, $\Cut$, and $\Shf$, see Proposition~\ref{prop:deriv},
  we get $\Sigma \proves A \Rightarrow B$.
  
  Conversely, by Theorem~\ref{th:commut} we get that $A \Rightarrow B$
  is provable by a finite $\Sigma' \subseteq \Sigma^\PL$ using only $\Ref$
  and $\Cut$. Therefore, we may form the \itm{1} and \itm{2}-parts
  of the derivation sequence using the formulas in $\Sigma'$ followed
  by $A \Rightarrow A$. Next, observe that there is a finite sequence
  $A_0,\ldots,A_n$ of subsets of $\alltatr$ such that $A_0 = A$,
  $A_i  = A_{i-1} \cup F$ for some $E \Rightarrow F \in \Sigma'$
  satisfying $E \subseteq A_{i-1}$, and $A_n \supseteq B$. In order
  to see that, consider~\eqref{eqn:A_n+1} and the fact
  that $A \Rightarrow B$ is provable by $\Sigma'$ without using $\Shf$.
  By moment's reflection, we can see that the \itm{4}-part of the
  derivation sequence is formed of formulas $A \Rightarrow A_i$
  ($i=0,\ldots,n$), and the sequence is terminated by
  a single application of $\Pro$ to obtain $A \Rightarrow B$.
\end{proof}

We conclude the section by showing further properties of provability. The
next assertion may be viewed as a type of a deduction theorem.

\begin{theorem}
  Let $\Sigma$ be a theory and $A,B \subseteq \alltatr$ be finite.
  Then the following statements are equivalent:
  \begin{enumerate}\parskip=0pt%
  \item[\itm{1}]
    $\Sigma \cup \{\emptyset \Rightarrow A\} \proves \emptyset \Rightarrow B$,
  \item[\itm{2}]
    there are $i_1,\ldots,i_n \in \Z$ such that
    $\Sigma \proves \textstyle\bigcup_{m=1}^n(A+i_m) \Rightarrow B$.
  \end{enumerate}
\end{theorem}
\begin{proof}
  ``\itm{1}\,$\Rightarrow$\,\itm{2}'':
  Let $A_1 \Rightarrow B_1,\ldots,A_n \Rightarrow B_n$ be a proof
  of $\emptyset \Rightarrow B$ by $\Sigma \cup \{\emptyset \Rightarrow A\}$.
  For each $p=1,\ldots,n$, we show that there are $i_1,\ldots,i_{p_n} \in \Z$
  for which $\Sigma \proves
  A_p \cup \textstyle\bigcup_{m=1}^{p_n}(A+i_m) \Rightarrow B_p$.
  The proof goes by induction
  on $p$. Thus, take $p=1,\ldots,n$ and assume the claim holds for all $q < p$.
  We distinguish the following cases:
  \begin{itemize}\parskip=0pt%
  \item[--]
    $A_p \Rightarrow B_p$ is an instance of $\Ax$.
    Then, we let $p_n = 1$, $i_1 = 0$, and thus
    $A_p \cup \textstyle\bigcup_{m=1}^{p_n}(A+i_m)$ equals $A_p \cup A$, i.e.,
    $A_p \cup A \Rightarrow B_p$ follows using $\Ax$.
  \item[--]
    $A_p \Rightarrow B_p \in \Sigma$. As in the previous case,
    for $p_n = 1$ and $i_1=0$ using $\Wea$ we infer
    $A_p \cup A \Rightarrow B_p$, showing 
    $\Sigma \proves A_p \cup A \Rightarrow B_p$.
  \item[--]
    Let $A_p \Rightarrow B_p$ result by
    $A_q \Rightarrow B_q$ and 
    $A_r \Rightarrow B_r$ using $\Cut$. In this case,
    there is $C$ such that $A_r = B_q \cup C$, $B_p = B_r$,
    and $A_p = A_q \cup C$.
    By induction hypothesis, there are $i_1,\ldots,i_{q_n} \in \Z$
    and $i'_1,\ldots,i'_{q_r} \in \Z$ such that
    $\Sigma \proves
    A_q \cup \textstyle\bigcup_{m=1}^{q_n}(A+i_m) \Rightarrow B_q$ and
    $\Sigma \proves
    B_q \cup C \cup \textstyle\bigcup_{m=1}^{q_r}(A+i'_m) \Rightarrow B_r$.
    Therefore, using $\Cut$,
    $\Sigma \proves
    A_q \cup \textstyle\bigcup_{m=1}^{q_n}(A+i_m)
    \cup C \cup \textstyle\bigcup_{m=1}^{q_r}(A+i'_m) \Rightarrow B_p$.
    Hence, for $i''_1 = i_1,\ldots,i''_{q_n}=i_{q_n},
    i''_{q_{n+1}}=i'_1,\ldots,i''_{q_n+q_r}=i'_{q_r}$ it follows that
    $\Sigma \proves
    A_q \cup C \cup \textstyle\bigcup_{m=1}^{q_n+q_r}(A+i''_m)
    \Rightarrow B_p$, i.e., 
    $\Sigma \proves
    A_p \cup \textstyle\bigcup_{m=1}^{q_n+q_r}(A+i''_m) \Rightarrow B_p$.
  \item[--]
    Let $A_p \Rightarrow B_p$ result by $A_q \Rightarrow B_q$ using $\Shf$.
    Then, $A_p = A_q + i$ and $B_p = B_q + i$ for some $i \in \Z$.
    By induction hypotheses, there are $i_1,\ldots,i_{q_n}$ such that
    $\Sigma \proves
    A_q \cup \textstyle\bigcup_{m=1}^{q_n}(A+i_m) \Rightarrow B_q$.
    Using $\Shf$, we get 
    $\Sigma \proves
    \bigl(A_q \cup \textstyle\bigcup_{m=1}^{q_n}(A+i_m)\bigr)+i
    \Rightarrow B_q+i$.
    Now, observe that
    $\bigl(A_q \cup \textstyle\bigcup_{m=1}^{q_n}(A+i_m)\bigr)+i$
    equals
    $(A_q+i) \cup \textstyle\bigcup_{m=1}^{q_n}(A+i_m+i)$.
    Therefore, the claim holds for integers $i_1+i,\ldots,i_{q_n}+i$.
  \end{itemize}
  As a special case for $p = n$, we get \itm{2} because $A_n = \emptyset$.
  
  ``\itm{2}\,$\Rightarrow$\,\itm{1}'':
  Let $\Sigma \proves \textstyle\bigcup_{m=1}^n(A+i_m) \Rightarrow B$ for
  some $i_1,\ldots,i_n \in \Z$. From the monotony of provability, we get
  that $\Sigma \cup \{\emptyset \Rightarrow A\} \proves
  \textstyle\bigcup_{m=1}^n(A+i_m) \Rightarrow B$.
  Moreover, for each $m=1,\ldots,n$ we get 
  $\Sigma \cup \{\emptyset \Rightarrow A\} \proves
  \emptyset \Rightarrow A+i_m$ using $\Shf$. Hence,
  $\Sigma \cup \{\emptyset \Rightarrow A\} \proves
  \emptyset \Rightarrow \textstyle\bigcup_{m=1}^nA+i_m$
  by finitely many applications of $\Add$ and $\Tra$ gives
  $\Sigma \cup \{\emptyset \Rightarrow A\} \proves \emptyset \Rightarrow B$. 
\end{proof}

\begin{example}
  Let us observe that a direct counterpart of the classic deduction theorem
  does not hold in our system. For instance, we may take a theory
  $\Sigma = \{\emptyset \Rightarrow \{x^1\}\}$. Then, using $\Shf$ for $i = 1$,
  we easily see that $\Sigma \proves \emptyset \Rightarrow \{x^2\}$. On the
  other hand, we have $\nproves \{x^1\} \Rightarrow \{x^2\}$ and thus in general
  $\Sigma \cup \{\emptyset \Rightarrow A\} \proves \emptyset \Rightarrow B$
  does not imply that $\Sigma \proves A \Rightarrow B$ which holds in the
  classic case.
\end{example}

\begin{example}\label{ex:completion}
  One of the classic laws about provability that apply to attribute implications
  and can be formulated in terms of attribute implications as formulas with
  limited expressive power compared to general propositional formulas is the
  principle of the proof by cases. Formally, if $\mathcal{R}$ consists only
  of $\Ax$ and $\Cut$, then the following are equivalent:
  \begin{itemize}
  \item
    $\Sigma \proves_\mathcal{R} A \Rightarrow B$;
  \item
    $\Sigma \cup \{C \Rightarrow D\} \proves_\mathcal{R} A \Rightarrow B$
    and
    $\Sigma \cup \{D \Rightarrow C\} \proves_\mathcal{R} A \Rightarrow B$.
  \end{itemize}
  This follows immediately by the fact that in this case,
  $\proves_\mathcal{R}$ becomes the classic propositional provability. The law
  does not apply in our system where $\mathcal{R}$ contains
  the additional rule $\Shf$. For instance, consider the following theory
  \begin{align*}
    \Sigma &= \{\{x^0\} \Rightarrow \{c^1\},
             \{x^0\} \Rightarrow \{d^2\},
             \{c^2\} \Rightarrow \{y^0\},
             \{d^1\} \Rightarrow \{y^0\}\}.
  \end{align*}
  Obviously, we have 
  $\Sigma \cup \{\{c^0\} \Rightarrow \{d^0\}\} \proves
  \{x^0\} \Rightarrow \{y^0\}$ using $\Shf$ and two applications of $\Cut$.
  Analogously, we get $\Sigma \cup \{\{d^0\} \Rightarrow \{c^0\}\} \proves
  \{x^0\} \Rightarrow \{y^0\}$. On the other hand, we can show that
  $\Sigma \nproves \{x^0\} \Rightarrow \{y^0\}$, i.e., the principle of the
  proof by cases does not hold. In order to see that 
  $\Sigma \nproves \{x^0\} \Rightarrow \{y^0\}$, observe that 
  $[\{x^0\}]_\Sigma = \{y^{-1}, x^0, c^1, y^1, d^2\}$ for which 
  $[\{x^0\}]_\Sigma \nmodels \{x^0\} \Rightarrow \{y^0\}$. Thus,
  since our logic is sound and $[\{x^0\}]_\Sigma \in \mathrm{Mod}(\Sigma)$,
  we indeed have $\Sigma \nproves \{x^0\} \Rightarrow \{y^0\}$.
\end{example}

\begin{remark}
  We may say that $\Sigma'$ is a \emph{completion} of $\Sigma$
  if $\Sigma \subseteq \Sigma'$ and for any finite $C,D \subseteq \alltatr$,
  we have either $\Sigma' \vdash C \Rightarrow D$ or 
  $\Sigma' \vdash D \Rightarrow C$. Let us note that analogous notions of
  completions play an important role in completeness proofs of
  various logics, cf.~\cite{Haj:MFL}. Namely, if a given theory does
  not prove a formula it is often desirable to find its completion which
  does not prove the formula as well. As a consequence of
  Example~\ref{ex:completion}, we observe that this is not possible in
  our logic. Namely, the example shows a particular case where 
  $\Sigma \nproves \{x^0\} \Rightarrow \{y^0\}$ and there is
  no completion $\Sigma'$ such that
  $\Sigma' \nproves \{x^0\} \Rightarrow \{y^0\}$.
  Indeed, each completion $\Sigma'$ proves either 
  $\{c^0\} \Rightarrow \{d^0\}$ or $\{d^0\} \Rightarrow \{c^0\}$
  and thus it also proves $\{x^0\} \Rightarrow \{y^0\}$. Nevertheless,
  we were able to prove Theorem~\ref{th:scmpl} without having this property.
\end{remark}

%%%%%%%%%%%%%%%%%%%%%%%%%%%%%%%%%%%%%%%%%%%%%%%%%%%%%%%%%%%%%%%%%%%%%%%%%%%%%%%%
%%%%% COMPUTATIONAL ISSUES
%%%%%%%%%%%%%%%%%%%%%%%%%%%%%%%%%%%%%%%%%%%%%%%%%%%%%%%%%%%%%%%%%%%%%%%%%%%%%%%%
\section{Computational Issues}\label{sec:comput}
In this section, we show bounds on the computational complexity of deciding
whether an attribute implication over attributes annotated by time points
is provable by a finite set $\Sigma$ of other attribute implications.
Then, we focus on a subproblem which typically appears in applications.
For the subproblem we provide a pseudo-polynomial time~\cite{GaJo:CaI}
decision algorithm.

We formalize the \emph{decision problem of entailment} as a language of
encodings of finitely many formulas, i.e., we put
\begin{align}
  L_\mathrm{ENT} &= \{\langle\Sigma,A \Rightarrow B\rangle \,|\,
                   \Sigma \text{ is a finite theory and } 
                   \Sigma\proves A\Rightarrow B\},
\end{align}
considering a fixed $\alltatr$. 
In order to show the lower bound of the time complexity of 
$L_\mathrm{ENT}$, we utilize a reduction of decision problems~\cite{Papa}
which involves the unbounded subset sum problem. The decision variant of the 
unbounded subset sum problem is formulated as follows: An instance of the 
problem is given by $n$ non-negative integers $j_1,\dots,j_n$ and a target 
value $z$; the answer to the instance is ``yes'' if{}f there are 
non-negative integers $c_1,\ldots,c_n$ such that
\begin{align}
  \textstyle\sum_{i=1}^{n} c_i j_i &= z.
  \label{eqn:uss}
\end{align}
The unbounded subset sum decision problem is
NP-complete, see~\cite[Proposition A.4.1]{Knapsack}. 

Let us note that in the case of the ordinary
attribute implications and functional dependencies, the problem of determining
whether a given formula follows by a finite set of formulas is easy and there
exist efficient linear time decision algorithms~\cite{BeBe:Cprttdonfrs}.
In contrast, the corresponding decision problem in our setting is hard:

\begin{theorem}[lower bound]\label{th:ENT}
  $L_\mathrm{ENT}$ is NP-hard.
\end{theorem}
\begin{proof}
  We prove the claim by showing that the unbounded subset sum problem
  (see Section~\ref{sec:prelim}) is polynomial time reducible to
  $L_\mathrm{ENT}$. Consider an instance of the unbounded subset sum problem
  given by non-negative integers $j_1,\dots,j_n$ and $z$. For the integers
  we consider
  \begin{align}
    \Sigma = \bigl\{\{y^0\}\Rightarrow\{y^{j_i}\} \,|\, i=1,\ldots,n\bigr\}
  \end{align}
  and put $A = \{y^0\}$, $B = \{y^{z}\}$. We now prove that 
  $\sum_{i=1}^{n} c_i j_i = z$ holds true for some non-negative integers
  $c_1,\ldots,c_n$ if{}f
  $\Sigma \proves \{y^0\} \Rightarrow \{y^z\}$ by proving both implications.
  
  In order to prove the if-part, assume that
  $\Sigma \proves \{y^0\} \Rightarrow \{y^z\}$.
  Using Theorem~\ref{th:norm}, it follows there is a normalized derivation
  sequence $\varphi_1,\dots,\varphi_k$ of
  $\{y^0\} \Rightarrow \{y^z\}$ using formulas in $\Sigma$.
  In the proof, we utilize a part of the sequence which results by
  applications of $\Acc$, see Definition~\ref{def:normseq}\,\itm{4}.
  All formulas in this part of the sequence can be written as
  \begin{align*}
    \underbrace{\{y^0\} \Rightarrow A_i}_{\varphi_i},
    \underbrace{\{y^0\} \Rightarrow A_{i+1}}_{\varphi_{i+1}}, \dots,
    \underbrace{\{y^0\} \Rightarrow A_{k-1}}_{\varphi_{k-1}},
  \end{align*}
  where $A_i,\dots,A_{k-1}$ are finite subsets of $\alltatr$,
  $A_i=\{y^0\}$, and $y^z \in A_{k-1}$ because
  $\varphi_k$ results from $\varphi_{k-1}$ by $\Pro$,
  cf. Definition~\ref{def:normseq}. By induction, we show for
  every $A_l$ ($i \leq l \leq k-1$) that the following property is
  satisfied:
  \begin{quote}
    If $y^w \in A_l$, then there are non-negative integers
    $c_1,\ldots,c_n$ \\ such that $w = \sum_{i=1}^{n} c_i j_i$.
  \end{quote}
  Notice the property is satisfied for $l = i$ since in that case
  we have $A_l = A_i = \{y^0\}$ and thus, we may put 
  $c_1 = c_2 = \cdots = c_n = 0$. Assuming the claim holds for
  $l$, we prove it for $l+1$ as follows. Inspecting 
  Definition~\ref{def:normseq}\,\itm{4}, it follows that
  $\{y^0\} \Rightarrow A_{l+1}$ results from $\{y^0\} \Rightarrow A_l$ and
  $\{y^0\}+t \Rightarrow \{y^{j_m}\}+t$ using $\Acc$ where $t \in \Z$ and 
  $1 \leq m \leq n$. As a consequence $\{y^0\}+t\subseteq A_l$ and thus,
  by induction hypothesis, there are non-negative integers $d_1,\ldots,d_n$
  such that $t = 0+t = \sum_{i=1}^{n} d_i j_i$. Then,
  $j_m + t = j_m + \sum_{i=1}^{n} d_i j_i$ and so
  $j_m + t = \sum_{i=1}^{n} c_i j_i$ for 
  non-negative integers $c_1,\ldots,c_n$ defined by
  \begin{align*}
    c_i &=
          \left\{
          \begin{array}{@{\,}l@{\quad}l@{}}
            d_i + 1, &\text{if } i = m, \\
            d_i, &\text{otherwise.}
          \end{array}
                   \right.
  \end{align*}
  Now, since we have $A_{l+1} \subseteq A_l \cup\{y^{j_m+t}\}$, the
  property holds for $A_{l+1}$. As a particular case, for
  $\{y^z\}\subseteq A_{k-1}$ we conclude there are non-negative
  integers $c_1,\ldots,c_n$ for which $\sum_{i=1}^{n} c_i j_i = z$
  which concludes the first part of the proof of Theorem~\ref{th:ENT}.
  
  Conversely, let $\sum_{i=1}^{n} c_i j_i = z$
  for some non-negative integers $c_1,\ldots,c_n$. By induction,
  we show that $\Sigma \proves \{y^0\} \Rightarrow \{y^{z_k}\}$ for
  every $z_k = \sum_{i=1}^{k} c_i j_i$ where $k=0,\ldots,n$.
  As a particular case for $k = n$, we obtain the desired fact that
  $\Sigma \proves \{y^0\} \Rightarrow \{y^z\}$ because $z_n = z$.
  
  Observe that for $k=0$, the claim follows trivially by $\Ax$. Now,
  suppose the claim holds for $k < n$. By induction hypothesis,
  $\Sigma \proves \{y^0\} \Rightarrow \{y^{z_k}\}$. Moreover,
  we have $\Sigma \proves \{y^0\} \Rightarrow \{y^{j_{k+1}}\}$
  because $\{y^0\} \Rightarrow \{y^{j_{k+1}}\} \in \Sigma$.
  Using $\Shf$, we also get
  $\Sigma \proves \{y^0\} + j_{k+1} \Rightarrow \{y^{j_{k+1}}\} + j_{k+1}$,
  i.e., using $\Cut$, it follows that 
  $\Sigma \proves \{y^0\} \Rightarrow \{y^{2j_{k+1}}\}$.
  Repeating the last argument $c_{k+1}$-times, we obtain
  $\Sigma \proves \{y^0\} \Rightarrow \{y^{c_{k+1}j_{k+1}}\}$.
  Now, using \Shf, we get
  $\Sigma \proves \{y^0\} + z_k \Rightarrow \{y^{c_{k+1}j_{k+1}}\} + z_k$, i.e.,
  $\Sigma \proves \{y^{z_k}\} \Rightarrow \{y^{c_{k+1}j_{k+1} + z_k}\}$.
  Hence, $\Sigma \proves \{y^0\} \Rightarrow \{y^{z_{k+1}}\}$
  follows by $\Cut$ using the fact that $z_{k+1} = z_k + c_{k+1}j_{k+1}$,
  which finishes the proof.
\end{proof}

The reduction involved in Theorem~\ref{th:ENT} is
illustrated in the following example.

\begin{example}
  Let us show a particular instance of the unbounded subset sum problem and
  its reduction to $L_\mathrm{ENT}$. We consider integers $5$, $7$, $11$,
  and a target number $31$ as an instance of the problem. The answer to this 
  instance is ``yes'' because for numbers $4$, $0$, and $1$,
  the sum $4 \cdot 5 + 0 \cdot 7 + 1 \cdot 11$ is equal to $31$. 
  The corresponding theory $\Sigma$, see the proof of Theorem~\ref{th:ENT},
  is
  \begin{align*}
    \Sigma = \{
    \{y^0\}\Rightarrow\{y^5\},
    \{y^0\}\Rightarrow\{y^7\},
    \{y^0\}\Rightarrow\{y^{11}\}\}.
  \end{align*}
  In this case, $\{y^0\}\Rightarrow\{y^{31}\}$ is provable from $\Sigma$
  because we may chain four shifted instances of $\{y^0\} \Rightarrow \{y^5\}$
  and a single shifted instance of $\{y^0\} \Rightarrow \{y^{11}\}$ by
  using $\Cut$. It corresponds with the sum
  $4 \cdot 5 + 0 \cdot 7 + 1 \cdot 11$.
  In a more detail, the corresponding proof
  of $\{y^0\}\Rightarrow\{y^{31}\}$ by $\Sigma$
  is the following sequence of formulas:
  \bgroup%
  \def\Proof#1#2{\item\makebox[11em][l]{\ensuremath{#1}}#2}
  \begin{enumerate}\parskip=0pt
    \Proof{\{y^0\}\Rightarrow\{y^{5}\}}{formula in $\Sigma$}
    \Proof{\{y^0\}+5\Rightarrow\{y^{5}\}+5}{using $\Shf$ on 1.}
    \Proof{\{y^0\}\Rightarrow\{y^{10}\}}{using $\Cut$ on 1. and 2.}
    \Proof{\{y^0\}+10\Rightarrow\{y^{5}\}+10}{using $\Shf$ on 1.}
    \Proof{\{y^0\}\Rightarrow\{y^{15}\}}{using $\Cut$ on 3. and 4.}
    \Proof{\{y^0\}+15\Rightarrow\{y^{5}\}+15}{using $\Shf$ on 1.}
    \Proof{\{y^0\}\Rightarrow\{y^{20}\}}{using $\Cut$ on 5. and 6.}
    \Proof{\{y^0\}\Rightarrow\{y^{11}\}}{formula in $\Sigma$}
    \Proof{\{y^0\}+20\Rightarrow\{y^{11}\}+20}{using $\Shf$ on 8.}
    \Proof{\{y^0\}\Rightarrow\{y^{31}\}}{using $\Cut$ on 7. and 9.}
  \end{enumerate}%
  \egroup%
\end{example}

\begin{remark}
  The entailment problem is closely related to the existence of
  non-negative solutions of linear Diophantine equations.
  Indeed, for a theory $\Sigma$ which consists of formulas of the form
  $\{y^0\} \Rightarrow \{y^{j_i}\}$ for $i=1,\ldots,n$,
  by inspecting the proof of Theorem~\ref{th:ENT}, we can see that
  $\Sigma \vdash \{y^0\} \Rightarrow \{y^z\}$ if{}f the
  linear Diophantine equation $j_1x_1 + \cdots + j_nx_n = z$ has
  a non-negative solution.
\end{remark}

Our observations on the upper bound of computational complexity involve
additional classes of decision problems. In order to establish an upper bound,
we utilize the fact that the satisfiability problem of temporal logic
with ``until'' and ``since'' operators over a linear flow of time is decidable
in polynomial space~\cite{Re:tcodpfltl}. For the purpose of our proof,
we use the linear temporal logic over $\langle\Z,<\rangle$ with the unary temporal
operators $\boxast$ (always), $\circ_F$ (next time), and 
$\circ_P$ (previous time) because these operators are definable
using operators ``until'' and ``since'', see~\cite{ArKoRyZa:ccfltl} for details.

From now on, we consider $Y$ (the set of attributes) as (a subset of) the set
of propositional variables. Recall that formulas of the temporal logic with
the above-mentioned operators are defined as follows: Each $y \in Y$ is a formula;
if $\varphi$ and $\psi$ are formulas, then
$\neg \varphi$, $\varphi \logand \psi$, $\varphi \Rightarrow \psi$,
$\mathop{\boxast} \varphi$, $\mathop{\circ_F} \varphi$,
and $\mathop{\circ_P} \varphi$ are formulas.
In order to interpret the formulas we consider
a standard structure $\mathbf{K} = \langle W,e,r\rangle$ where
$W = \mathbb{Z}$, $r$ is the genuine ordering $<$ on $\Z$, and $e$ is
an evaluation such that $e(w,y)\in\{0,1\}$ for all $w \in \Z$ and $y \in Y$.
Given $\mathbf{K}$ and $w \in \Z$, we interpret the formulas as usual:
We put
\begin{enumerate}\parskip=0pt%
\item[\itm{1}]
  $\mathbf{K},w \models y$ whenever $e(w,y) = 1$;
\item[\itm{2}]
  $\mathbf{K},w \models \neg\varphi$ whenever
  $\mathbf{K},w \nmodels \varphi$;
\item[\itm{3}]
  $\mathbf{K},w \models \varphi \logand \psi$ whenever
  $\mathbf{K},w \models \varphi$ and $\mathbf{K},w \models \psi$;
\item[\itm{4}]
  $\mathbf{K},w \models \varphi \Rightarrow \psi$ whenever
  $\mathbf{K},w \nmodels \varphi$ or $\mathbf{K},w \models \psi$;
\item[\itm{5}]
  $\mathbf{K},w \models \mathop{\boxast} \varphi$ whenever
  $\mathbf{K},w' \models \varphi$ for all $w' \in \Z$;
\item[\itm{6}]
  $\mathbf{K},w \models \mathop{\circ_F} \varphi$ whenever
  $\mathbf{K},w' \models \varphi$ for $w' \in \Z$
  such that $w < w'$ and there does not exist $z \in \Z$ 
  such that $w < z < w'$;
\item[\itm{7}]
  $\mathbf{K},w \models \mathop{\circ_P} \varphi$ whenever
  $\mathbf{K},w' \models \varphi$ for $w' \in \Z$
  such that $w' < w$ and there does not exist $z \in \Z$ 
  such that $w' < z < w$.
\end{enumerate}
We say that $\varphi$ is \emph{true in $\mathbf{K}$} whenever
$\mathbf{K},w \models \varphi$ for all $w \in \Z$.
Moreover, we say that $\varphi$ is \emph{satisfiable} whenever there is a structure
$\mathbf{K}$ such that $\mathbf{K},0\models \varphi$. Moreover for each
formula of the form \eqref{eqn:cimpl}, we consider its counterpart
in the considered temporal logic 
\begin{align}
  \boxast\bigl(
  \bigl(\mdl{i_1}y_1 \logand \cdots \logand \mdl{i_m}y_m\bigr)
  \Rightarrow
  \bigl(\mdl{j_1}z_1 \logand \cdots \logand \mdl{j_n}z_n\bigr)
  \bigr),
  \label{eqn:mimpl}
\end{align}
where $\mdl{i}$ is defined as follows:
\begin{align}
  \mdl{i}y 
  &=
  \left\{
    \begin{array}{@{\,}l@{\quad}l@{}}
      y, & \text{if } i = 0, \\
      \circ_F\mdl{i-1}y, & \text{if } i > 0, \\
      \circ_P\mdl{i+1}y, & \text{if } i < 0. \\
    \end{array}
  \right.
\end{align}
Note that the construction of $\mdl{i}y$ from $y^i$ requires space which
is linear in (the absolute value of) $i \in \Z$, i.e., it is exponential
in the length of the encoding of~$i$.

\begin{theorem}\label{th:expred}
  $L_\mathrm{ENT}$ is reducible in exponential space to the satisfiability
  problem of the linear temporal logic over $\langle\Z,<\rangle$
  with unary temporal operators ``always'', ``next time'',
  and ``previous time''.
\end{theorem}

\begin{proof}
  First, observe that for each subset of $\alltatr$ we may consider
  a corresponding structure which makes the same formulas
  true---any $A \Rightarrow B$ is true in the subset of $\alltatr$
  if{}f its counterpart given by~\eqref{eqn:mimpl} is true in 
  the corresponding structure.
  Namely, for $M \subseteq \alltatr$, we may consider
  $\mathbf{K}_{M} = \langle W,e,r\rangle$, where $e(w,y) = 1$
  if $y^w \in M$ and $e(w,y) = 0$ otherwise. Conversely,
  for $\mathbf{K} = \langle W,e,r\rangle$,
  we put $M_\mathbf{K} = \{y^w \,|\, e(w,y) = 1\}$.
  Now, for any $w\in W$, it is easy to see that $M \models A \Rightarrow B$ 
  if{}f	$\mathbf{K}_M,w \models \varphi$ where $\varphi$ is the
  counterpart to $A \Rightarrow B$ given by~\eqref{eqn:mimpl}.
  From now on, we tacitly identify attribute implications with their
  counterparts. Furthermore, we have
  $\mathbf{K},w \models A \Rightarrow B$ if{}f
  $M_\mathbf{K} \models A \Rightarrow B$.

  Now, for a given $\Sigma = \{A_1\Rightarrow B_1,\dots,A_m\Rightarrow B_m\}$
  and $A\Rightarrow B$ we may consider formula
  $A_1\Rightarrow B_1 \logand \cdots \logand A_m\Rightarrow B_m 
  \logand \neg \bigl( A\Rightarrow B\bigr)$ whose construction requires
  exponential space. From the previous observation, it is obvious that
  the formula is satisfiable if{}f $\Sigma \nmodels A \Rightarrow B$.
\end{proof}

\begin{corollary}[upper bound]\label{col:expspace}
  $L_\mathrm{ENT}$ belongs to EXPSPACE.
\end{corollary}
\begin{proof}
  The decision procedure reduces the input of $L_\mathrm{ENT}$ to 
  the satisfiability problem of linear temporal logic over $\langle\Z,<\rangle$
  with unary temporal operators ``always'', ``next time'', and ``previous time''
  in exponential space, see Theorem~\ref{th:expred}. Then, the input is reduced
  to the satisfiability problem of the linear temporal logic over
  $\langle\Z,<\rangle$ with binary temporal operators ``until'' and ``since''
  in linear space~\cite{ArKoRyZa:ccfltl} which we can decide in polynomial
  space~\cite{Re:tcodpfltl}. Altogether, the decision procedure
  decides $L_\mathrm{ENT}$ in exponential space.
\end{proof}

\begin{remark}
  Note that the results of Theorem~\ref{th:expred} and
  Corollary~\ref{col:expspace} can also be interpreted so that 
  $L_\mathrm{ENT}$ is decidable in a \emph{pseudo-polynomial space}
  because we reduce an instance of $L_\mathrm{ENT}$ to an instance
  (of the satisfiability problem of the above-mentioned temporal logic)
  the length of which is bounded from above by the numeric value encoded
  in the original input. With respect to the new instance, the decision
  procedure works in polynomial space.
\end{remark}

We now turn our attention to issues of entailment of formulas which typically
appear in applications in prediction. The restriction on particular
formulas allows us to improve the complexity of the entailment problem.
Based on the time points present in antecedents and consequents of
attribute implications, we may consider formulas
which describe presence of attributes in future time points. That is, based on
the presence of attributes in the past, the formulas indicate which attributes
are present in future time points. Technically, such formulas can be seen
as attribute implications where all the time points in the antecedents are
smaller (i.e., denote earlier time points) than all the time points in the
consequents which denote later time points. We call such formulas predictive
and define the notion as follows.

\begin{definition}\label{def:reg}
  An attribute implication $A \Rightarrow B$ over $Y$ annotated by time points
  in $\Z$ is called \emph{predictive} whenever $A \ne \emptyset$,
  $B \ne \emptyset$, and for each $x^i \in A$ and $y^j \in B$,
  we have $i \leq j$. A theory $\Sigma$ is called predictive
  whenever all its formulas are predictive.
\end{definition}

\begin{remark}
  Note that the deduction rules (Shf) and (Cut) preserve the property of
  being predictive. That is, if $A \Rightarrow B$ is provable by
  a predictive theory $\Sigma$ without using (Ax),
  then $A \Rightarrow B$ is predictive. General instances of (Ax)
  are not predictive formulas.
\end{remark}

In the next assertion, we utilize lower and upper time bounds of finite
non-empty subsets of $\alltatr$:
For a finite non-empty $M \subseteq \alltatr$, put
\begin{align}
  l(M) &= \min\{i \in \Z \,|\, y^i \in M\}, \\
  u(M) &= \max\{i \in \Z \,|\, y^i \in M\}.
\end{align}
Thus, $l(M)$ and $u(M)$ are the lowest and greatest time points
which appear in~$M$, respectively. Clearly, $A \Rightarrow B$ is 
\emph{predictive} if{}f both $A$ and $B$ are non-empty and $u(A) \leq l(B)$.

\begin{theorem}\label{th:predbound}
  Let $\Sigma$ and $A \Rightarrow B$ be predictive. Then, for
  \begin{align}\label{eqn:SigmaAB}
    \Sigma_{\!A}^B 
    &=
      \{E+i \Rightarrow F+i \,|\, 
      E \Rightarrow F \in \Sigma \text{ and }
      l(A)-l(E) \leq i \leq u(B) - l(F)\}
  \end{align}
  we have $\Sigma \vdash A \Rightarrow B$ if{}f
  $\Sigma_{\!A}^B \vdash_\mathcal{R} A \Rightarrow B$ for
  $\mathcal{R}$ containing $\Ax$ and $\Cut$.
\end{theorem}
\begin{proof}
  Observe that the if-part of the claim is trivial. In order to prove the
  only-if part, assume that $\Sigma \vdash A \Rightarrow B$. That is,
  $B \subseteq [A]_\Sigma$ owing to Theorem~\ref{le:synclos} and
  Theorem~\ref{th:semclos}.
  Note that $\Sigma_{\!A}^B \vdash_\mathcal{R} A \Rightarrow B$ for $\mathcal{R}$
  containing $\Ax$ and $\Cut$ means that $A \Rightarrow B$ is provable by
  $\Sigma_{\!A}^B$ as an ordinary attribute implication.
  Let $A^\circ$ denote the least subset of
  $\alltatr$ with the following properties:
  \begin{enumerate}\parskip=0pt%
  \item[\itm{1}]
    $A \subseteq A^\circ$, and
  \item[\itm{2}]
    for each $E \Rightarrow F \in \Sigma_{\!A}^B$:
    if $E \subseteq A^\circ$ then $F \subseteq A^\circ$.
  \end{enumerate}
  Since $A^\circ$ is in fact the syntactic closure of $A$ with respect
  to $\mathcal{R}$, $\Sigma_{\!A}^B \vdash_\mathcal{R} A \Rightarrow B$ if{}f
  $B \subseteq A^\circ$. That is, in order to prove the desired claim,
  it suffices to show that $A^\circ \cap T = [A]_\Sigma \cap T$ for
  \begin{align*}
    T = \{y^i \in \alltatr \,|\, l(A) \leq i \leq u(B)\}.
  \end{align*}
  Trivially, we get that
  $A^\circ \cap T \subseteq [A]_\Sigma \cap T$. In order to prove the
  converse inclusion, according to Theorem~\ref{th:clos}, it suffices to
  check that $A_\Sigma^n \cap T \subseteq A^\circ \cap T$ for each
  non-negative integer $n$. By induction, assume that
  $A_\Sigma^n \cap T \subseteq A^\circ \cap T$ and take
  $y^j \in (A_\Sigma^{n+1} \cap T) \setminus (A_\Sigma^{n} \cap T) =
  (A_\Sigma^{n+1} \setminus A_\Sigma^{n}) \cap T$.
  The fact $y^j \in A_\Sigma^{n+1} \setminus A_\Sigma^{n}$ yields there
  is $E \Rightarrow F \in \Sigma$ and $i \in \Z$ such that
  $E+i \subseteq A_\Sigma^{n}$ and $y^j \in F+i$. It can be shown
  that $E+i \Rightarrow F+i \in \Sigma_{\!A}^B$.
  Indeed, since $\Sigma$ is predictive,
  observe that $l(E)+i = l(E+i) \geq l(A_\Sigma^n) = l(A)$
  and thus $i \geq l(A) - l(E)$.
  Moreover, $y^j \in F+i$ yields $l(F+i) = l(F) + i \leq j$ and thus
  $i \leq j - l(F)$ which gives $i \leq u(B) - l(F)$ on
  account of $j \leq u(B)$ since $y^j \in T$. As a consequence,
  $E+i \Rightarrow F+i \in \Sigma_{\!A}^B$. Furthermore,
  $E+i \subseteq A_\Sigma^{n}$ and the fact that $E \Rightarrow F$
  is predictive give
  $E+i = (E+i) \cap T \subseteq A_\Sigma^{n} \cap T$.
  By induction hypothesis, $E + i \subseteq A^\circ$ and thus
  $F + i \subseteq A^\circ$ by \itm{2}. Hence, $y^j \in A^\circ$
  and so $A_\Sigma^{n+1} \cap T \subseteq A^\circ \cap T$.
\end{proof}

Let $L_\mathrm{PRE}$ be the language consisting of encodings of pairs of
all finite predictive theories and predictive formulas, i.e.,
\begin{align}
  L_\mathrm{PRE} &= \{\langle\Sigma,A \Rightarrow B\rangle \,|\,
                   \Sigma \text{ is finite and }
                   \Sigma \text{ and } A \Rightarrow B \text{ are predictive}\}.
\end{align}
Based on Theorem~\ref{th:predbound}, we establish the following
observation on the time complexity of deciding whether a predictive
formula is provable by a finite predictive theory.

\begin{theorem}
  $L_\mathrm{ENT} \cap L_\mathrm{PRE}$ is decidable in a pseudo-polynomial time.
\end{theorem}
\begin{proof}
  Take a finite predictive $\Sigma$ and a predictive formula $A \Rightarrow B$.
  The theory $\Sigma_{\!A}^B$ given by~\eqref{eqn:SigmaAB} is finite.
  According to
  Theorem~\ref{th:predbound}, the problem of deciding
  $\Sigma \vdash A \Rightarrow B$ is reducible to the problem of
  deciding whether $\Sigma_{\!A}^B$ entails $A \Rightarrow B$
  without using $\Shf$,
  i.e., in the sense of the entailment of ordinary attribute implications.
  Therefore, the problem is decidable in a time which is polynomial with
  respect to the size of
  $\Sigma_{\!A}^B$~\cite{BeBe:Cprttdonfrs,GaWi:FCA,Mai:TRD}.
  Now, observe that the size of (the encoding of) $\Sigma_{\!A}^B$
  may be bounded from above by the size of (the encoding of) $\Sigma$
  multiplied by
  \begin{align}
    n = \textstyle\max\{
    \max(0, u(B) + l(E) - l(A) - l(F) + 1) \,|\,
    E \Rightarrow F \in \Sigma\},
    \label{eqn:pseudolenght}
  \end{align}
  i.e., the size of $\Sigma_{\!A}^B$ is polynomial in the numeric value
  encoded in the input $\Sigma$ and hence $L_\mathrm{ENT} \cap L_\mathrm{PRE}$
  is decidable in a pseudo-polynomial time.
\end{proof}

\begin{remark}
  (a)
  By considering only $L_\mathrm{ENT} \cap L_\mathrm{PRE}$, we have
  improved the upper bound since pseudo-polynomial time algorithms belong
  to EXPTIME~\cite{GaJo:CaI} which is believed to be better than EXPSPACE.
  Observe that $L_\mathrm{ENT} \cap L_\mathrm{PRE}$ is also NP-hard because
  we can use the same reduction as in Theorem~\ref{th:ENT}.

  (b)
  Because of the complexity issues, in applications it is reasonable to
  consider attribute implications annotated by time points with small
  difference between lower and upper time bounds
  (\emph{maxspan}~\cite{FDL:Iarmcptasmd})
  since $L_\mathrm{ENT} \cap L_\mathrm{PRE}$ is decidable
  in \emph{pseudo-linear time} with respect to $n$
  given by~\eqref{eqn:pseudolenght}.
\end{remark}	

\begin{algorithm}[p]
  \caption{\rmfamily\textsc{PseudoLinClosure}\,$(\Sigma, A, \mathit{Max})$}%
  \label{alg:pseudolinclosure}
  \SetKw{KwSet}{set}
  \SetKw{KwAdd}{add}
  \SetKw{KwChoose}{choose}
  \SetKw{KwFrom}{from}
  \def\count#1#2{\ensuremath{\mathop{\text{\textit{count}}}[#1,#2]}}
  \def\list#1{\ensuremath{\mathop{\text{\textit{list}}}[#1]}}
  \def\update{\ensuremath{\mathop{\text{\textit{update}}}}}
  \def\new{\ensuremath{\mathop{\text{\textit{new}}}}}
  \ForAll{$E \Rightarrow F \in \Sigma$}{%
    \For{$i$ \KwFrom $l(A) - l(E)$ \KwTo $\mathit{Max} - l(F)$}{%
      \KwSet $\count{E \Rightarrow F}{i}$ \KwTo $|E|$\;
      \ForAll{$y^j \in E$}{%
        \KwAdd $\langle E \Rightarrow F,i\rangle$ \KwTo $\list{y^{i+j}}$\;
      }
    }
  }
  \KwSet $M$ \KwTo $A$\;
  \KwSet $\update$ \KwTo $A$\;		\label{alg:endinit}
  \While{$\update \ne \emptyset$}{%
    \KwChoose $y^i$ \KwFrom $\update$\;
    \KwSet $\update$ \KwTo $\update \setminus \{y^i\}$\;
    \ForAll{$\langle E \Rightarrow F, j\rangle \in \list{y^i}$}{%
      \KwSet $\count{E \Rightarrow F}{j}$
      \KwTo $\count{E \Rightarrow F}{j} - 1$\;
      \If{$\count{E \Rightarrow F}{j} = 0$}{%
        \KwSet $\new$ \KwTo $F+j \setminus M$\;
        \KwSet $M$ \KwTo $M \cup \new$\;
        \KwSet $\update$ \KwTo $\update \cup \new$\;
      }
    }
  }
  \Return{M}
\end{algorithm}%

An explicit procedure for deciding $L_\mathrm{ENT} \cap L_\mathrm{PRE}$ in
a \emph{pseudo-linear} time is
described in Algorithm~\ref{alg:pseudolinclosure}. It is a generalization of
\textsc{LinClosure}~\cite{BeBe:Cprttdonfrs}, cf. also~\cite{Mai:TRD},
which incorporates applicable time shifts of formulas in $\Sigma$.
The algorithm accepts three arguments:
\begin{enumerate}
\item
  a finite predictive theory $\Sigma$,
\item
  a finite $A \subseteq \alltatr$, and
\item
  a non-negative number $\mathit{Max} \geq u(A)$,
\end{enumerate}
and it returns a subset $M \subseteq [A]_\Sigma$ such that
$M \cap T = [A]_\Sigma \cap T$ for
\begin{align}
  T = \{y^i \in \alltatr \,|\, l(A) \leq i \leq \mathit{Max}\}.
  \label{eqn:T_Max}
\end{align}
The soundness of the algorithm is justified by the following observation:

\begin{theorem}\label{th:pseudolc}
  Let $\Sigma$ and $A \Rightarrow B$ be predictive and
  let $\Sigma$ be finite. Then, Algorithm~\ref{alg:pseudolinclosure} executed
  with arguments $\Sigma$, $A$, and $u(B)$, terminates after finitely many
  steps and for the returned value $M$ we have $\Sigma \proves A \Rightarrow B$
  if{}f $B \subseteq M$.
\end{theorem}
\begin{proof}
  The arguments are fully analogous to those in case of the classic
  \textsc{LinClosure}, so we present here comments on issues arising
  only in the context of attributes annotated by time points. Technical details
  can be found in~\cite{BeBe:Cprttdonfrs}.
  Notice that Algorithm~\ref{alg:pseudolinclosure} uses auxiliary structure
  \emph{count} and \emph{list} to store information about formulas.
  The structure \emph{count} can be seen as an associative array indexed
  by (pointers to) formulas in $\Sigma$ and integers $i$ representing
  time shifts. The value of $\mathop{\mathit{count}}[E \Rightarrow F,i]$
  is initially set to the number of attributes in the antecedent of
  $E \Rightarrow F$ (shifted by $i$). During the computation,
  $\mathop{\mathit{count}}[E \Rightarrow F,i]$ represents the number of
  remaining attributes in $E+i$ which have not been ``updated.''
  The structure \emph{list} is an array indexed by attributes annotated by
  time points and the value of $\mathop{\mathit{list}}[y^i]$ is a list of
  records $\langle E \Rightarrow F, j\rangle$ representing (pointers to)
  formulas in $\Sigma$ and their $j$-shifts such that $y^i$ appears in
  the antecedent of $E \Rightarrow F$ shifted by $j$. An additional variable
  \emph{update} is initialized at line 10 and maintains attributes annotated
  by time points which are waiting to be ``updated.'' An update of $y^i$,
  see lines 13--21, consists in decrementing the counter of occurrences
  of attributes in shifted antecedents in all formulas where $y^i$ appears.
  All such formulas (and their $j$-shifts) are found
  in $\mathop{\mathit{list}}[y^i]$, see line 14.
  If $\mathop{\mathit{count}}[E \Rightarrow F,j]$ reaches zero,
  see line 16, the antecedent of $E+j \Rightarrow F+j$ is already contained
  in $M$, and all new attributes in $F+j$ are prepared for update.
  Clearly, the procedure terminates after finitely many steps,
  and by Theorem~\ref{th:predbound},
  the attributes annotated by time points accumulated in $M$
  represent a subset of $[A]_\Sigma$. In addition,
  if $u(B) \leq \mathit{Max}$, then $B \subseteq M$ if{}f
  $B \subseteq [A]_\Sigma$ if{}f $\Sigma \proves A \Rightarrow B$
  as a consequence of our previous observations.
\end{proof}

\begin{remark}
  The procedure in Algorithm~\ref{alg:pseudolinclosure} is called
  \textsc{PseudoLinClosure} because for given parameters,
  $\Sigma$, $A$, and $\mathit{Max}$, it computes a subset of the
  closure of $[A]_\Sigma$ in a linear time with respect to the
  numeric value of the encoding of its input arguments, i.e.,
  its time complexity is \emph{pseudo-linear.} Indeed, this is
  a consequence of the fact that each $y^i$ where
  $l(A) \leq i \leq \mathit{Max}$ is updated during the computation
  at most once.
\end{remark}

\begin{example}
  \def\atr#1#2{\ensuremath{\mathtt{\lowercase{#1#2}}}}%
  Consider a set $M$ given by the table in Figure \ref{fig:weather}.
  Since $M$ can be regarded as transactional data over a set of
  items $Y$ with a dimensional attribute $\mathfrak{d}$ the domain of which
  is $\Z$, we can utilize the algorithm proposed in~\cite{LFH:biaammiar}.
  The parameters for the algorithm are numbers \emph{maxspan}, \emph{minsupport},
  and \emph{minconfidence} for which we obtain a set $\Sigma$ of all
  predictive $A \Rightarrow B$ where $u(A\cup B) - l(A\cup B) \leq
  \text{\emph{maxspan}}$,
  $\text{\emph{minconfidence}} \leq \text{\emph{confidence}}(A\Rightarrow B)$,
  and $\text{\emph{minsupport}} \leq \text{\emph{support}}(A\Rightarrow B)$.
  For this particular example we consider $\text{\emph{maxspan}} = 5$,
  $\text{\emph{minconfidence}} = 1$ since we are interested in
  formulas true in $M$, and $\text{\emph{support}} = 5$. In this setting,
  we obtain
  \begin{align*}
    \Sigma = \{ &\{\atr{W}{m}^0\}\Rightarrow\{\atr{T}{c}^4\},	
    \{\atr{W}{l}^0\}\Rightarrow\{\atr{T}{c}^3\}, \\
    &\{\atr{W}{l}^0\}\Rightarrow\{\atr{W}{m}^1\},	
    \{\atr{W}{l}^0\}\Rightarrow\{\atr{W}{m}^1, \atr{T}{c}^3\}, \\
    &\{\atr{W}{l}^0, \atr{W}{m}^1\}\Rightarrow\{\atr{T}{c}^3\}, 
    \{\atr{R}{n}^0, \atr{W}{l}^2\}\Rightarrow\{\atr{T}{c}^3\}, \\
    &\{\atr{R}{n}^0, \atr{R}{n}^3\}\Rightarrow\{\atr{T}{c}^3\}, 
    \{\atr{T}{c}^0, \atr{R}{n}^5\}\Rightarrow\{\atr{T}{c}^5\}, \\
    &\{\atr{T}{c}^0, \atr{T}{c}^3, \atr{R}{n}^5\} \Rightarrow\{\atr{T}{c}^5\}, 
    \{\atr{R}{n}^0, \atr{T}{c}^0, \atr{R}{n}^3\} \Rightarrow\{\atr{T}{c}^3\}, \\
    &\{\atr{R}{n}^0, \atr{T}{c}^0, \atr{W}{m}^2\} \Rightarrow\{\atr{T}{c}^3\}\}.
  \end{align*}
  Now, we may successively reduce the set $\Sigma$ by removing formulas 
  $A\Rightarrow B$ such that     
  $\Sigma \setminus \{A\Rightarrow B\}\proves A\Rightarrow B$, i.e., without loss 
  of information. Since $\Sigma$ is predictive we may use
  \textsc{PseudoLinClosure} and obtain the following set:
  \begin{align*}
    \Sigma' = \{&\{\atr{W}{m}^0\}\Rightarrow\{\atr{T}{c}^4\},	
    \{\atr{W}{l}^0\}\Rightarrow\{\atr{W}{m}^1, \atr{T}{c}^3\}, \\ 
    &\{\atr{R}{n}^0, \atr{R}{n}^3\}\Rightarrow\{\atr{T}{c}^3\}, 
    \{\atr{R}{n}^0, \atr{W}{m}^2\}\Rightarrow\{\atr{T}{c}^3\}, \\
    &\{\atr{T}{c}^0, \atr{R}{n}^5\}\Rightarrow\{\atr{T}{c}^5\}\},
  \end{align*}
  i.e., the equivalent non-redundant set contains less than half
  of the formulas in~$\Sigma$. For $\text{\emph{maxspan}} = 5$ and
  $\text{\emph{support}} = 2$, the reduction is much more significant.
  From the total number of $34,440$ generated formulas,
  \textsc{PseudoLinClosure} can be used to produce an equivalent set
  consisting of only $81$ formulas.
\end{example}

%%%%%%%%%%%%%%%%%%%%%%%%%%%%%%%%%%%%%%%%%%%%%%%%%%%%%%%%%%%%%%%%%%%%%%%%%%%%%%%%
%%%%% CONCLUSION
%%%%%%%%%%%%%%%%%%%%%%%%%%%%%%%%%%%%%%%%%%%%%%%%%%%%%%%%%%%%%%%%%%%%%%%%%%%%%%%%
\section{Conclusion}\label{sec:concl}
We have presented logic for reasoning with if-then rules expressing
dependencies between attributes changing in time. The logic extends
the classic logic for dealing with if-then rules by considering discrete
time points as an additional component.
We have studied both the semantic entailment based on
preserving validity in models in all time points and syntactic entailment
represented by a provability relation. We have shown a characterization
of the semantic entailment based on least models and syntactico-semantical
completeness of the logic. We have shown the problem of entailment is NP-hard,
decidable in exponential space, and its simplified variant which involves
only predictive formulas is decidable in pseudo-linear time.
Future research directions we consider
interesting include utilization of generalized
quantifiers~\cite{Li:FOPLwGQ,Mo:Ogq} to capture notions like
``validity in all time points with possible exceptions'',
connections to rules which may emerge in temporal databases~\cite{DDL:TRT},
and further analysis of algorithms related to the entailment.

%%%%%%%%%%%%%%%%%%%%%%%%%%%%%%%%%%%%%%%%%%%%%%%%%%%%%%%%%%%%%%%%%%%%%%%%%%%%%%% 
\subsection*{Acknowledgment}
Supported by grant no. \verb|P202/14-11585S| of the Czech Science Foundation.
J.~Triska was also supported by internal student grant IGA\_PrF\_2015\_023
of Palacky University Olomouc.

%%%%%%%%%%%%%%%%%%%%%%%%%%%%%%%%%%%%%%%%%%%%%%%%%%%%%%%%%%%%%%%%%%%%%%%%%%%%%%%%
%%%%% REFERENCES
%%%%%%%%%%%%%%%%%%%%%%%%%%%%%%%%%%%%%%%%%%%%%%%%%%%%%%%%%%%%%%%%%%%%%%%%%%%%%%%%

\footnotesize
\bibliographystyle{amsplain}
\bibliography{jtvv}

% that's all folks
\end{document}